\documentclass[12pt, draftclsnofoot, onecolumn]{IEEEtran}

\normalsize
\usepackage{amsmath, amsfonts, amssymb, amsthm}

\usepackage{optidef}
\usepackage{color}
\usepackage{bm}
\usepackage{tikz,pgfplots}
\usepackage{algorithm}
\usepackage{cite}
\usepackage{algpseudocode}
\usepackage{siunitx}
\usepackage{subcaption}
\usepackage{standalone}
\usepackage{enumitem}

\usepackage{accents}

\newcommand{\iif}[1] {\ \mathbf{1} {\big({#1}\big)}  }
\newcommand{\eexp}[1] {\mathbb{E}\left\{{#1}\right\}}

\DeclareMathOperator*{\argmax}{arg\,max}
\DeclareMathOperator*{\argmin}{arg\,min}

\def \alg {EECW}
\def \du {\Delta U(t)}
\def \de {\Delta E(t)}
\def \dum {\Delta U_{max}}
\def \dem {\Delta E_{max}}
\def \oot {{1 \over 2}}
\def \C {\mathcal{C}}
\def \Ut {U_0}
\def \Am {A_{\text{m}}}
\def \Mumax {\mu_{\text{m}}}
\def \Pam {P_{\text{AP}}^{\text{max}}}
\def \Pumax{P_{\text{um}}}
\def \Pm {P_{\text{WN}}^{\text{max}}}

\def \No {N_o}
\def \Ni {N_i}
\def \Cm {C_{\text{m}}}

\def \uo {\mu^{\text{out}}_{n,s}(t)}
\def \ui {\mu^{\text{in}}_{n,s}(t)}
\def \uio {\mu^{\text{i-o}}_{n,s}(t)}
\def \eo {\phi^{\text{out}}_n(t)}
\def \ei {\phi^{\text{in}}_n(t)}
\def \eio {\phi^{\text{i-o}}_n(t)}
\def \eitau {\phi^{\text{in}}_n(\tau)}
\def \Eh {E_n(t)}

\def \uimax {\mu_{\text{max}}}

\def \uitau {\mu^\text{in}_{n,s}(\tau)}

\def \uos {\mu^\text{out}_{n,s}}
\def \uis {\mu^\text{in}_{n,s}}
\def \uios {\mu^{i-o}_{n,s}}

\def \eos {\phi^\text{out}_n}
\def \eis {\phi^\text{in}_n}
\def \eios {\phi^{i-o}_n}

\def \eimax {\phi_{\text{max}}}

\def \Rnode {N_t(l)}
\def \Tnode {N_h(l)}

\def \Td {\tau_d(t)}
\def \Te {\tau_e(t)}
\def \Tf {\tau_f}

\def \Uc {\mathcal{U}_c}
\def \Ec {\mathcal{E}_c}

\def \Out {\mathcal{O}}
\def \In {\mathcal{I}}

\newtheoremstyle{note}
{3pt}
{3pt}
{}
{}
{\bfseries\upshape}
{.}
{.5em}
{}

\theoremstyle{note}
\newtheorem{lemma}{Lemma}
\newtheorem{theorem}{Theorem}

\usetikzlibrary{arrows}
\usetikzlibrary{arrows,decorations.pathmorphing,backgrounds,positioning,fit,petri,shapes.geometric}
\title{An Energy-Efficient  Controller for Wirelessly-Powered Communication Networks }
	
\author{ \IEEEauthorblockN{Mohammad Movahednasab, Behrooz Makki, \textit{Member}, \textit{IEEE}, Naeimeh Omidvar, \textit{Member}, \textit{IEEE} Mohammad Reza Pakravan, \textit{Member}, \textit{IEEE}, Tommy Svensson, \textit{Senior Member}, \textit{IEEE} and Michele Zorzi, \textit{Fellow}, \textit{IEEE}}\thanks{Part of this work has been accepted for  presentation at ICC 2019.}}

\begin{document}

\maketitle

\begin{abstract}
 In a wirelessly-powered communication network (WPCN), an energy access point (E-AP) supplies the energy needs of the network nodes through radio frequency wave transmission, and  the nodes store their received energy in their batteries for  possible data transmission. In this paper, we propose an online control policy for energy transfer from the E-AP to the wireless nodes and for data transfer among the nodes. With our proposed control policy, all  data queues of the nodes are stable, while the average energy consumption of the network is shown to be within a bounded gap of the minimum energy  required for stabilizing the network.  Our  proposed policy is designed using  a quadratic Lyapunov function to capture the limitations on the energy consumption of the nodes imposed by their battery levels.  We show that under the proposed control policy, the backlog level in the data queues and the stored energy level in the batteries fluctuate in small intervals around some constant levels.
 Consequently, by imposing negligible average data drop rate,   the data buffer size and the battery capacity of the nodes can be significantly reduced. 
\end{abstract}

\section{Introduction}
Smart electronic devices are increasingly making their way into our daily life. It is predicted that by 2021, there will be around 28 billion connected devices  all over the world \cite{Ab2016}, a great number of which will be portable and battery-powered. However, in some applications such as biomedical implants inside human bodies \cite{Zeng2016} or distributed monitoring  sensors, replacing the batteries may be infeasible. As such, the problem of providing the required energy for the portable battery-operated devices has recently received growing attention, both in academia and industry \cite{Huang2015,Zeng2016}.  Particularly, the idea of charging batteries over the air is considered as a promising solution which guarantees an uninterrupted connection and  reduces the problem of massive battery disposal.   The key enabling technology for charging over the air is  wireless energy transfer (WET). There are various WET methods including electromagnetic radiation \cite{Huang2015}, resonant coupling \cite{Resonant2009} and inductive coupling \cite{Inductive2013}. Compared to the two latter methods, electromagnetic radiation provides a wider coverage range and is more flexible for transmitter/receiver deployment and movement  \cite{Zeng2016}.  

 There are numerous studies on energy beamforming as a technique for alleviating the high transmission path loss in   wirelessly-powered communication networks (WPCNs)  \cite{Yang2014, Liu2014,Nguyen2017,Lau2014} as well as for the simultaneous wireless information and power transfer systems  \cite{Larsson2013,Nasir2016, Shi2014}. Moreover, \cite{Ju2014,Bafghi2017,Haghifam2016,Haghifam2016_2,2016arXiv161105995H,Liu2016,Xu2017,Gurakan2016} consider  cooperation among the users as a useful method to increase the network coverage in two-hop \cite{Ju2014,Bafghi2017,Haghifam2016,Haghifam2016_2,2016arXiv161105995H,Liu2016} and  multi-hop  \cite{Xu2017,Gurakan2016} WPCNs, respectively.
Furthermore, \cite{Behrooz2016} and \cite{Reliable2} study the reliability of data transmission in WPCNs.

 In the networks that support  continuous or regular communication, the nodes are equipped with batteries which enable {them} to store their harvested energy in one time-slot for possible use in the subsequent time-slots \cite{Biason2017,Biason2015,Choi2015,Choi2018}.  In  such cases,  the network performance should be analyzed in the {long term}, because  a single time-slot analysis may not be optimal in general. For this reason, \cite{Biason2017} and \cite{Biason2015}  study the long-term network throughput optimization through Markov decision processes (MDP).  Moreover, in \cite{Choi2015,Choi2018,RezaeiPIMRC,RezaeiGlobe,RezaeiArxiv},  long-term energy optimality and fairness for a multi-user downlink WET scenario are studied through Lyapunov optimization technique.

In this work, we design an energy-efficient WET policy that jointly controls data-link power allocation, data routing, energy beamforming and  data/energy  transmission time sharing in a multi-hop WPCN.  The problem is cast in the form of minimizing the total average  energy consumption of the network subject to stability of the data queues in the network and the battery level constraints of the nodes.  The battery level  constraint complicates  finding the optimal control policy, since high energy consumption {in} one time-slot  degrades the battery level of the node considerably, which may lead to  energy outage in the subsequent time-slots.  Therefore, the optimal decisions {in} different time-slots are  coupled.  This coupling makes finding the optimal policy  challenging.

We use  Lyapunov optimization method with a novel quadratic Lyapunov function  to avoid energy outage. Based on the proposed Lyapunov function, we  propose an online control policy called  energy-efficient controller for WPCN (\alg{}), that does not require the explicit knowledge of the channel statistics. With the proposed policy, the  time-slots are  devoted to either energy transfer or data transmission.  In  energy transmission time-slots, the energy beam is focused towards the nodes with low battery levels, higher queue backlogs and  higher energy-link channel gains. In data transmission time-slots, the data is routed through the  nodes with less congested queues and higher battery levels.  

We  analyze the performance of the proposed control policy and  prove that for every arbitrarily chosen value of a parameter $V > 0$ in our algorithm, the energy consumption under  \alg{}  is within  a bounded gap of the order of   $\mathcal{O}({1\over V})$ to the optimal policy, while the average backlog of data queues is upper bounded by $\mathcal{O}(V)$.  In addition, we show that the backlog level of the data queues and the energy level of the batteries    converge probabilistically to some constant values, with the probability of deviation from those values decreasing exponentially with respect to the amount of deviation. Using this result, we further propose a   modified version of \alg{} which can significantly reduce the required size of the data buffers in the nodes as well as the required capacity of their batteries, while  imposing a negligible  drop rate in the network. Finally, we present extensive simulations to provide insightful intuitions on the advantages of the proposed control policy, in terms of its energy requirements for stabilizing a WPCN.

As opposed to \cite{Yang2014, Liu2014,Nguyen2017,Lau2014,Larsson2013,Nasir2016,Shi2014,Ju2014,Bafghi2017,Haghifam2016,Haghifam2016_2,2016arXiv161105995H,Liu2016,Xu2017,Gurakan2016}, we  consider battery-powered nodes and analyze the network performance in the {long term}, instead of a single time-slot analysis.  In contrast to   \cite{Yang2014, Liu2014,Nguyen2017,Larsson2013,Nasir2016,Shi2014,Ju2014,Bafghi2017,Haghifam2016,Haghifam2016_2,2016arXiv161105995H,Liu2016,Xu2017,Gurakan2016,RezaeiPIMRC,RezaeiGlobe,RezaeiArxiv,Biason2015,Biason2017,Behrooz2016,Reliable2},
 {which analyze the throughput, the delay or the outage probability, we study the optimization of the energy consumption while the data queues are stabilized.}  Finally, different form  \cite{Behrooz2016,Reliable2,Yang2014,Liu2014,Nguyen2017,Lau2014,Larsson2013,Nasir2016,Choi2015,Choi2018,RezaeiPIMRC,RezaeiGlobe,RezaeiArxiv,Biason2015,Biason2017}, we study a general multi-hop WPCN where the data should be routed through the nodes. The differences in the system model and the problem formulation makes our results and  analysis completely different from those in \cite{Yang2014, Liu2014,Nguyen2017,Lau2014,Larsson2013,Nasir2016,Shi2014,Ju2014,Bafghi2017,Haghifam2016,Haghifam2016_2,2016arXiv161105995H,Liu2016,Xu2017,Gurakan2016,Choi2015,Choi2018,RezaeiPIMRC,RezaeiGlobe,RezaeiArxiv,Biason2015,Biason2017,Behrooz2016,Reliable2}.

The rest of the paper is organized as follows.  The considered system model and our problem formulation are illustrated in Section \ref{sec:sysMo}. {The proposed control policy as well as its performance analysis is presented in Section \ref{sec:policy}}.     The behavior of the data backlog in the queues and the energy level in the batteries are analyzed in Section \ref{sec:behavior}. Some implementation issues are discussed in Section \ref{sec:Implementation}. Simulation results are presented in Section \ref{sec:sim}, and finally, Section \ref{sec:conclude} concludes the paper. 

Notation: Matrices and vectors are denoted by small and capital boldface letters, respectively. Moreover, unless otherwise mentioned, vectors are single row-matrices. Also, $(.)^T$, $(.)^H$ and $(.)^\ast$   denote  transpose,   conjugate transpose and  element-wise conjugate of a matrix, respectively.  Finally, $|.|$ denotes the absolute value (or the modulus for  complex numbers),  $\lVert.\rVert$ denotes the norm of  vectors, $\mathbb{E}\{.\}$ represents the expectation and  $[x]^{+}=\max\{x,0\},\;\forall x \in \mathbb{R}$.





\section{System Model}\label{sec:sysMo}

\begin{table}

\caption{Notation Summary.}
\vspace{-0.50cm}
\begin{center}
\resizebox{0.95\textwidth}{!}{
\begin{tabular}{  |c|c|}

\hline
Symbol                  				& Definition   \\ \hline
$N, S, L$                       			& Number of nodes, streams and data links, respectively. \\\hline
$U_{n,s}(t)$		 				& The backlog of data queue allocated to stream $s$ at node $n$ in time-slot $t$.\\\hline
$ B_n(t)$               				& The battery level of node $n$ in time-slot $t$. \\\hline
$\Tnode,\Rnode$ 		              &The head  and tail node of  link $l$, respectively. \\\hline
$ \Out_n, \In_n$              			&The set of outgoing links from and  incoming links  to node $n$, respectively. \\\hline
$\ei, \eo  $               	       			& The energy stored in and drained from  the battery of node $n$ in time-slot $t$, respectively.\\\hline
$\eimax $               	       		       & The limitation on $\ei$ and $\eo$ (i.e., $\ei, \eo \leq \eimax\;\; \forall t$) \\\hline
$ \bm{w}(t), P_{\text{AP}}(t)$						& The beamforming vector and the transmission power of the E-AP in time-slot $t$, respectively. \\ \hline
$A_{n,s}(t)$               				& {The instantaneous data of stream $s$ arrived at node $n$, in time-slot $t$}.\\\hline
$\lambda_{n,s} $					&{The data arrival rate of stream $s$ at node $n$}.\\\hline
$\ui,\uo $							& The total amount of data of stream $s$ that enters to and exits from node $n$ in time-slot $t$, respectively.\\\hline
$\uimax $						& The limit on $\ui$ and $\uo$ (i.e., $\ui, \uo \leq \uimax\;\; \forall t$)\\\hline
$\bm{g}(t)= [g_1(t),\ldots,g_L(t)]$		& The vector of  data link channel states in time-slot $t$. \\ \hline
$\bm{h}_n(t) = [h_n^1(t),\dots,h_n^M(t)]$					& The vector of  energy link channel gains for node $n$ in time-slot $t$. \\ \hline
$\bm{p}(t)= [p_1(t),\ldots,p_l(t)]$						& The vector of power allocations to data links in time-slot $t$. \\ \hline

$\Pi$							& The set of feasible data-link power allocation vectors.\\ \hline
$R_l(\bm{p}(t), \bm{g}(t))$			& The rate-power function of link $l$ in time-slot $t$. \\ \hline
$R_{l,s}(t)$						& The instantaneous data rate of stream $s$ over link $l$ in time-slot $t$. \\ \hline
$\Pm, \Pam $	 				& Maximum admissible transmission power of the wireless nodes and the E-AP, respectively.\\ \hline
$\Tf, \Td,\Te$                                     & Time-slot duration and fraction of time-slot devoted to data and energy transmission,  respectively.\\\hline
\end{tabular}}
\end{center}
\vspace{-0.5cm}
\label{table:notation}
\end{table}

The notations used in the paper along with their definitions are presented in Table \ref{table:notation}. We consider a WPCN consisting of one energy access point (E-AP) and $N$ wireless nodes, with $S$ streams of data between distinct endpoints in the network. { It should be noted that our analysis can be extended to consider multiple E-APs. However, for simplicity  we focus on networks with a single E-AP.} The wireless nodes are battery-powered, and the batteries are recharged by the energy received from the E-AP. The E-AP is equipped with $M$ antennas to focus its transmission beam towards the nodes. Moreover, we assume that the nodes use a single antenna for both energy reception and data transmission/reception.
 There exist $N$ energy links between the E-AP and the nodes and $L$  data links between the nodes. The topology of a sample network  is depicted in Fig. \ref{fig:sampNetA}. 
For each data link $l\in\{1, \ldots, L\}$, $\Tnode$ and $\Rnode$ denote the head node and the tail node of   link $l$, respectively. Moreover, we define $\mathcal{I}_n$ and $\mathcal{O}_n$ as  the sets of the incoming and outgoing data links of node $n$, respectively. 
\begin{figure}
\centering
\begin{subfigure}{0.39\textwidth}
\centering
  \begin{tikzpicture}
   [WD/.style={circle,draw,fill=blue!10,scale =0.6},
   HAP/.style={rectangle ,draw,fill=red!10,scale =1},
     ST/.style={rectangle ,fill=none,scale =1},labelSt/.style={scale = 1},scale=2.5]

     \node(HAP) at (      0, 0) [HAP]{E-AP};
     \node(WD1) at (    -0.8, 0.6) [WD,fill=green!60] {1};
     \node(WD2) at ( -0.8, -0.25) [WD,fill=yellow!60] {2};
     \node(WD3) at ( -0.5, 0.25) [WD] {3};
     \node(WD4) at (     0, 0.5) [WD] {4};
     \node(WD5) at (  0.5, 0.25) [WD] {5};
     \node(WD6) at (     1, 1) [WD,fill=green!60] {6};
     \node(WD7) at ( -0.2, -0.5) [WD] {7};
     \node(WD8) at (    0.5, -0.25) [WD] {8};
     \node(WD9) at (    0.8,0) [WD,fill=yellow!60] {9};

 	  \draw [->,  shorten >=0.01cm,shorten <=0.01cm , line width = 0.3mm ] (WD1.east) -- (WD4.west) node[midway,above, labelSt] {\small1}; 		     
 	  \draw [->,  shorten >=0.01cm,shorten <=0.01cm , line width = 0.3mm] (WD1.south east) -- (WD3.north west) node[xshift=-0.16cm, yshift=-0.02cm, midway,  labelSt, rotate = 0 ] {\small 2};         
 	  \draw [->,  shorten >=0.01cm,shorten <=0.01cm , line width = 0.3mm] (WD3.north east) -- (WD4.south west) node[xshift=-0.09cm, yshift=0.35cm,midway, below,  labelSt, rotate = 0] {\small3}; 
 	  \draw [->,  shorten >=0.01cm,shorten <=0.01cm , line width = 0.3mm] (WD2.north east) -- (WD3.south west) node[xshift=-0.01cm, yshift=0.00cm, midway, right,  labelSt, rotate = 0] {\small4}; 
 	  \draw [->,  shorten >=0.01cm,shorten <=0.01cm , line width = 0.3mm] (WD3.east) -- (WD5.west)node[xshift=0.21cm, yshift=-0.09cm,midway, above,  labelSt, rotate = 0] {\small5};                    
 	  \draw [->,  shorten >=0.01cm,shorten <=0.01cm , line width = 0.3mm] ([xshift=-0.06cm, yshift=-0.03cm]WD3.south east) -- (WD7.north)node[xshift=0.15cm, yshift=-0.27cm,midway, left,  labelSt, rotate = 0] {\small6};         
 	  \draw [->,  shorten >=0.01cm,shorten <=0.01cm , line width = 0.3mm] (WD2.south east) -- (WD7.west)node[xshift=-0cm, yshift=-0.0cm, midway, below,  labelSt, rotate = 0] {\small7}; 
 	  \draw [->,  shorten >=0.01cm,shorten <=0.01cm , line width = 0.3mm] (WD7.east) -- (WD8.west)node[xshift=0.12cm, yshift=-0.2cm,midway, left,  labelSt, rotate = 0] {\small8};                    
 	  \draw [->,  shorten >=0.01cm,shorten <=0.01cm , line width = 0.3mm] (WD4.north east) -- (WD6.west)node[xshift=0.00cm, yshift=-0.02cm,midway, above,  labelSt, rotate = 0] {\small9};                    
 	  \draw [->,  shorten >=0.01cm,shorten <=0.01cm , line width = 0.3mm] (WD4.south east) -- (WD5.north west)node[xshift=-0.1cm, yshift=0.1cm,midway, right,  labelSt, rotate = 0] {\small10}; 
 	  \draw [->,  shorten >=0.01cm,shorten <=0.01cm , line width = 0.3mm] (WD5.north) -- (WD6.south west)node[xshift=-0.08cm, yshift=-0.13cm,midway, right,  labelSt, rotate = 0] {\small11}; 
 	  \draw [->,  shorten >=0.01cm,shorten <=0.01cm , line width = 0.3mm] (WD5.east) -- (WD9.north west)node[xshift=-0.2cm, yshift=0.15cm,midway, right,  labelSt, rotate = 0] {\small12};           
 	  \draw [->,  shorten >=0.01cm,shorten <=0.01cm , line width = 0.3mm] ( [xshift=0.03cm] WD5.south) -- ([xshift=0.03cm] WD8.north)node[xshift=-0.1cm, yshift=-0.15cm,midway, right,  labelSt, rotate = 0] {\small13};               
 	  \draw [->,  shorten >=0.01cm,shorten <=0.01cm , line width = 0.3mm]  ([xshift=-0.03cm]WD8.north) -- ([xshift=-0.03cm]WD5.south) node[xshift=0.03cm, yshift=0.1cm,midway, left,  labelSt, rotate = 0] {\small14};  
	  \draw [->,  shorten >=0.01cm,shorten <=0.01cm , line width = 0.3mm] (WD8.east) -- (WD9.south west)node[xshift=-0.03cm, yshift=-0.07cm,midway, right, labelSt, rotate = 0] {\small15};           
 	  \draw [->,  shorten >=0.01cm,shorten <=0.01cm , line width = 0.3mm]  ([xshift=-0.03cm]WD9.north) -- ([xshift=-0.03cm]WD6.south) node[xshift=0.03cm, yshift=+0.07cm,midway, right,  labelSt, rotate = 0] {\small16};           
 	  \draw [->,  shorten >=0.01cm,shorten <=0.01cm , line width = 0.3mm] ( [xshift=0.03cm]WD6.south) -- ( [xshift=0.03cm]WD9.north)node[xshift=-0.05cm, yshift=-0.15cm,midway, left,  labelSt, rotate = 0] {\small17};           

         \draw [->,  shorten >=0.1cm,shorten <=0.03cm ,dashed,line width = 0.3mm, color = red!80] ([xshift=0.5ex]HAP.north west) -- (WD1.east)node[xshift=-0.15cm, yshift=0.2cm, midway,  labelSt, rotate = 0] {}; 					 
         \draw [->,  shorten >=0.1cm,shorten <=0.03cm ,dashed,line width = 0.3mm, color = red!80] ([yshift = -0.1ex]HAP.west) -- (WD2.east)node[xshift=-0.15cm, yshift=-0.02cm, midway,  labelSt, rotate = 0] {}; 					 
         \draw [->,  shorten >=0.1cm,shorten <=0.03cm ,dashed,line width = 0.3mm, color = red!80] (HAP.west) -- (WD3.south east)node[xshift=-0.12cm, yshift=-0.02cm, midway,  labelSt, rotate = 0] {}; 					 
         \draw [->,  shorten >=0.1cm,shorten <=0.03cm ,dashed,line width = 0.3mm, color = red!80] (HAP.north) -- (WD4.south)node[xshift=-0.05cm, yshift=-0.09cm, midway,  labelSt, rotate = 0] {}; 					 
         \draw [->,  shorten >=0.1cm,shorten <=0.03cm ,dashed,line width = 0.3mm, color = red!80] (HAP.north east) -- (WD5.south west)node[xshift=-0.10cm, yshift=0.06cm, midway,  labelSt, rotate = 0] {}; 					 
 \draw [->,  shorten >=0.1cm,shorten <=0.03cm ,dashed,line width = 0.3mm, color = red!80] ([xshift = -0.4ex]HAP.north east) -- (WD6.south west)node[xshift=-0.12cm, yshift=-0.02cm, midway,  labelSt, rotate = 0] {}; 					 
\draw [->,  shorten >=0.1cm,shorten <=0.03cm ,dashed,line width = 0.3mm, color = red!80] (HAP.south) -- ([xshift = 0.8ex]WD7.north west)node[xshift=-0.12cm, yshift=-0.02cm, midway,  labelSt, rotate = 0] {}; 					 
\draw [->,  shorten >=0.1cm,shorten <=0.03cm ,dashed,line width = 0.3mm, color = red!80] (HAP.south east) -- (WD8.north west)node[xshift=-0.12cm, yshift=-0.02cm, midway,  labelSt, rotate = 0] {}; 					 
\draw [->,  shorten >=0.05cm,shorten <=0.03cm ,dashed,line width = 0.3mm, color = red!80] (HAP.east) -- (WD9.west)node[xshift=-0.12cm, yshift=0.05cm, midway, labelSt, rotate = 0] {}; 					 

\end{tikzpicture} 
\caption{An example of the network topology.}
\label{fig:sampNetA}
\end{subfigure}
~
\begin{subfigure}{0.39\textwidth}
\centering
  \begin{tikzpicture}
    [mynode/.style = {rectangle, draw, align=center,inner xsep=0.5mm, inner ysep=1mm},
    -latex,     
    ]

    \node[mynode] (q1) {$U_{4,1}$};
    \node[mynode, below = 0.1cm of  q1.south] (q2) {$U_{4,2}$};
    \node[mynode, below = 0.8cm of  q2.south, red] (e) {$B_4$};
\begin{scope}[on background layer]

    \node[fit=(q1) (q2)(e), draw, circle, label={\small node 4},fill = blue!10] (outer) {};
\end{scope}

    \draw[->,line width=2mm] (outer.north east) -- ++(8mm, 8 mm ) node[midway,above,xshift=+0.3cm, yshift=-0.7cm] {\small link 9};	  				
    \draw[->,line width=2mm] (outer.south east) -- ++(8mm, -8 mm )node[midway,above,xshift=+0.3cm, yshift=0.1cm] {\small link 10};					

    \draw[<-,line width=2mm] (outer.north west) -- ++(-8mm, 8 mm )node[midway,above,xshift=-0.6cm, yshift=-0.7cm] {\small link 1};					
    \draw[<-,line width=2mm] (outer.south west) -- ++(-8mm, -8 mm )node[midway,above,xshift=+0.1cm, yshift=-1.2cm] {\small$ R_3(\bm{p}(t),\bm{g}(t))$} node[midway,above,xshift=-0.6cm, yshift=0.1cm] {\small link 3};;		

    \draw[<-, dashed, line width=2mm, red] (e.270 ) -- ++(0, -12 mm )node[midway,above,xshift=+1.6cm, yshift=-0.8cm] {\small$P_{AP}(t) |\bm{w}(t)\bm{h}_4(t)|^2$}; 		

    \draw[->,densely dotted, line width=0.2mm] ([yshift = .5mm ]q1.east) to [out = 0, in = 225](outer.47);			
    \draw[->, densely dotted, line width=0.2mm]([yshift = -.5mm ]q1.east) to [out = 0, in = 135](outer.317) node[midway,above,xshift=1.5cm, yshift=-1.2cm,inner sep=0] {\small$R_{10,1}(t)$};			

    \draw[->, line width=0.2mm] ([yshift = 0.5mm ]q2.east) to [out = 0, in = 225]( outer.43) ;				
    \draw[->, line width=0.2mm] ([yshift = -0.5mm ]q2.east) to [out = 310, in = 135](outer.313) ;			

    \draw[->,densely dotted, line width=0.2mm]  (outer.133) to [out = 315, in = 170]([yshift = 0.5mm ]q1.west) ;			
    \draw[->,densely dotted, line width=0.2mm]  (outer.223) to [out = 45, in = 180] ([yshift = -0.5mm ]q1.west) ;			

    \draw[->, line width=0.2mm]   ( outer.137) to  [out = 315, in = 160 ]([yshift = 0.5mm ]q2.west);				
    \draw[->, line width=0.2mm]   (outer.227) to   [out = 40, in = 200] ([yshift = -0.5mm ]q2.west) node[midway,above,xshift=-0.23cm, yshift=-1.4cm, fill= blue!10, inner sep =0] {\small$ R_{3,2}(t)$}; 			


\end{tikzpicture}
\caption{An example of a wireless node.}
\label{fig:sampNode}
\end{subfigure}
\vspace{-0.2cm}
\caption{Sub-figure (a) shows an example  of the  network topology. The solid black and  dashed red arrows represent the data links and the energy links, respectively. In this example figure, there are two data streams  between  nodes 1 and 6 and  nodes 2 and 9,  i.e., the objective is to send the messages of node 1 (resp. 2) to node 6 (resp. 9). Sub-figure (b) shows the structure of node 4. It consists of two data queues and a battery. In this figure, $C_{10,1}(t)$ and $C_{3,2}(t)$   are the data rates assigned to stream 1 and stream 2 over  links 10 and 3, respectively. }\label{fig:sampNet}\vspace{-0.5cm}
\end{figure}
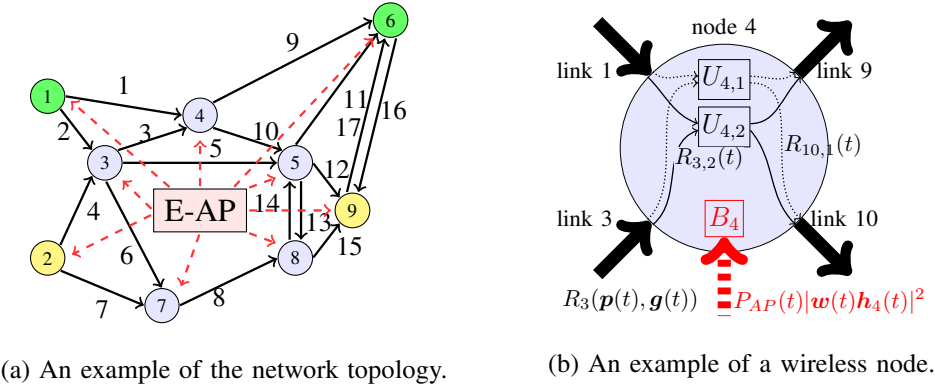
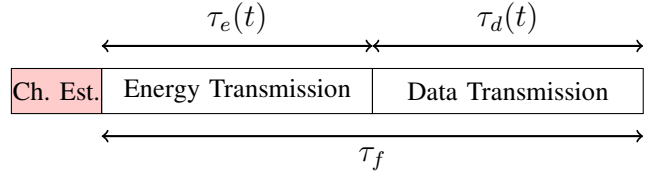
\begin{figure}
\centering
\begin{tikzpicture}[scale = 0.6]
\draw (-2,0) [fill=red!20] rectangle node{\small Ch. Est. } (0,1);
\draw (0,0)   rectangle node{\small Energy Transmission} (6,1);
\draw (6,0)  rectangle node{\small Data Transmission} (12,1);
\draw[<->,thick] (0,1.5) -- (6,1.5) node[midway,above] { $\Te$};
\draw[<->,thick] (6,1.5) -- (12,1.5) node[midway,above] { $\Td$};
\draw[<->,thick] (0,-0.5) -- (12,-0.5) node[midway,below] {$\Tf$};
\end{tikzpicture}
\vspace{-0.3cm}\caption{A time-slot structure.}\vspace{-0.7cm}
\label{fig:time}
\end{figure}

The time horizon is divided into time-slots with fixed length, indexed by $t$. Figure \ref{fig:time} shows the structure of a time-slot. At the beginning of each time-slot $t$, a small interval  is devoted to channel estimation and control signaling. The rest of the time-slot  is  divided into two intervals  of  lengths $\Te$ and $\Td$, for energy and data transmission, respectively.  We have $\Te+\Td = \Tf$ where $\Tf$ is the fixed  portion of the time-slot  allocated for data and energy transmission. 

The channel coefficients are assumed to be constant during a time-slot but vary randomly and independently in successive time-slots.  Recall that  the E-AP has multiple antennas whereas the wireless nodes use a single antenna for both energy reception and data transmission/reception. In each time-slot $t$, $g_l(t)$ and $h_n^m(t)$ represent the channel gains of the link between  nodes $\Tnode$ and $\Rnode$  and the link between the $m$-th antenna of the E-AP and  node $n$, respectively. Accordingly, we define  $\bm{g}(t) \triangleq (g_1(t),\ldots, g_L(t))$ and $\bm{h}_n(t) \triangleq (h_n^1(t), \ldots, h_n^M(t))$ as the channel gain vectors for data links and energy link of node $n$, respectively.  Note that here the energy links are numbered according to the ID of the   energy receiving node whereas  the data links are independently numbered.


\vspace{-0.2cm}
\begin{description}[leftmargin=0pt, listparindent=\parindent]
\item{\textbf{Data and Energy Transmission}}
Let $\bm{p}(t) \triangleq (p_1(t), \ldots,p_L(t))$ denote the data-link power vector, in which the $l$-th entry  is the allocated transmission power over the $l$-th data link. 
Moreover, let $\Pi$ denote the finite set of all feasible power vectors. We assume that setting an element of a power vector in $\Pi$ to zero results in a new power vector that also belongs to  $\Pi$. Furthermore, we assume that the maximum total transmit power of each node is limited to $\Pm$. Let $R_l(\bm{p}(t), \bm{g}(t))$ denote the rate-power function in link $l$ under the allocated data-link power vector $\bm{p}(t)$ and the channel gain vector $\bm{g}(t)$.  Consider two feasible power vectors $\bm{p}(t)$ and $\tilde{\bm{p}}(t)$, where $\tilde{p}_{l^\prime}(t)  < p_{l^\prime}(t) $ and $\tilde{\bm{p}}_{l}(t) = p_{l}(t),\forall l\neq {l^\prime}$. We assume that the rate-power functions  under each of these two power vectors satisfy the following properties
 \begin{align}
& \text{ if } \tilde{p}_{l^\prime}(t) =0 \text{ then }R_{l^\prime}(\tilde{\bm{p}}(t), \bm{g}(t)) = 0,\label{eq:noPowerNoRate}\\
& \exists \delta \geq 0: R_{l^\prime}(\bm{p}(t), \bm{g}(t)) - R_{l^\prime}(\tilde{\bm{p}}(t), \bm{g}(t)) \leq \delta (p_{l^\prime}(t)- \tilde{p}_{l^\prime}(t)),\label{eq:rateBound}\\
& R_l(\bm{p}(t), \bm{g}(t)) \leq R_{l}(\tilde{\bm{p}}(t), \bm{g}(t)) \;\; \forall {l} \neq {l^\prime}.\label{eq:interference}
\end{align}
{As an example of $R_l(\bm{p}(t), \bm{g}(t))$, the reader may think of 
\begin{align}\label{eq:sampleRatePower}
R_l(\bm{p}(t), \bm{g}(t))= \log\left(1 + { p_l(t)|g_l(t)|^2 \over N_0 + \sum_{l^{'} \in \In_{N_t(l)}/l} |g_{l^{'}}(t)|^2p_{l^{'}} (t)  }\right),
\end{align}
where $\In_{N_t(l)}/l$ is the set of the links that interfere with  link $l$.   Equation \eqref{eq:sampleRatePower} is an appropriate approximation for the achievable rates in the cases with codewords of moderate/large length.}  However, as we show in Section \ref{sec:sim}, our analysis is also applicable to the case of codewords of finite length, which results in a different rate-power function. Note that the  properties \eqref{eq:noPowerNoRate}, \eqref{eq:rateBound} and \eqref{eq:interference} are easily satisfied by conventional rate-power functions. Equation \eqref{eq:noPowerNoRate} indicates that no data can be passed through a data link if no power is assigned to that link.  Inequality \eqref{eq:rateBound} is satisfied by functions with bounded first derivative. Finally, inequality \eqref{eq:interference} holds due to the interference effect among wireless links. Let $R_{l,s}(t)$ denote the transmission rate allocated to stream $s$ in  link $l$. The sum rate of all streams in link $l$ should not exceed the achievable rate of that link. Therefore, a feasible rate allocation scheme should satisfy $\sum_{s=1}^SR_{l,s}(t) \leq R_l(\bm{p}(t), \bm{g}(t)).$

 The E-AP performs energy beamforming  to concentrate its transmit energy towards the nodes.  Let 
$\bm{w}(t) \triangleq (w_1(t), \ldots,w_M(t))$ denote the normalized beamforming vector of the E-AP.  Accordingly, the received energy at each node $n$, denoted by $\Eh$, is  given by
\vspace{-0.15cm}
\begin{align}
\Eh \triangleq \left|\bm{w}(t)\bm{h}_n^T(t)\right|^2 E_{\text{AP}}(t)\;\; \forall n,
\label{eq:recPower}
\vspace{-0.20cm}
\end{align}
where $E_{\text{AP}}(t) \triangleq \Te P_{\text{AP}}(t)$ is the E-AP's transmitted energy and  $P_{\text{AP}}(t)$ is the E-AP's transmit power in time-slot $t$. The peak transmission power of the E-AP is limited to  $\Pam$, i.e., $P_{\text{AP}}(t) \in [0,\Pam].$

%
%

%
%
%

\item{\textbf{Wireless Nodes} }
As shown in Fig.  \ref{fig:sampNode}, each node includes  $S$ data queues and a battery.  Let $U_{n,s}(t)$ denote the level of the stored data for stream $s$  in node $n$ at the beginning of  time-slot $t$. Moreover, let $\ui$ and $\uo$ denote the number of data units of stream $s$ that enter to and exit from node $n$ during time-slot $t$, respectively.  Accordingly,  $U_{n,s}(t)$ evolves as \vspace{-0.15cm}
\begin{align}
U_{n,s}(t+1) = \bigg [ U_{n,s}(t) - \uo\bigg]^{+}+\ui.
\label{eq:queueEvolve}
\end{align}
The  parameters $\uo$ and $\ui$ are determined through the assigned transmission rates and the external data arrivals, that is, 
\vspace{-0.25cm}
\begin{align}
\uo \triangleq \Td\sum_{l \in \mathcal{O}_n} R_{l,s}(t),
\vspace{-0.25cm}
\end{align}
and 
\vspace{-0.25cm}
\begin{align}
\ui \triangleq \Td\sum_{l \in \mathcal{I}_n} R_{l,s}(t) +A_{n,s}(t),
\end{align}
where $A_{n,s}(t) \in [0,\Am]$ is the number of external data units of stream $s$ that enter    node $n$ in time-slot $t$.    We assume that $A_{n,s}(t)$ is a random variable  following an identical and independent distribution  in different time-slots. We denote  the mean value of $A_{n,s}(t)$ by $\lambda_{n,s}$, and we have  $\lambda_{n,s} =\lambda_s$ if node $n$ is the source of stream $s$ and $\lambda_{n,s} =0$ otherwise. We denote   the vector of arrival rates by $\bm{\lambda} = [\lambda_1, \ldots,\lambda_S ] $.  Furthermore,   we assume that $\ui$ and $\uo$ are both upper bounded  by $\uimax$.
\color{black}

{The battery of each node  is recharged by the energy received from the E-AP and is (partially) discharged when the node transmits data. Let $B_n(t)$ denote the energy level stored in the battery of node $n$ at the beginning of time-slot $t$. Therefore, the battery level at node $n$ evolves according to \vspace{-0.15cm}}
%
\begin{align}\label{eq:batteryEvolve}
B_n(t+1) = B_n(t) - \eo+\ei,\vspace{-0.10cm}
\end{align}
where $\eo =  \Td\sum_{l \in \Out_n} p_l(t)$ is the total energy consumption of node $n$ in time-slot $t$ and $\ei$ is the portion of the received energy that is stored in the battery of node $n$ in time-slot $t$. 
Intuitively, we expect that a node stores all  its received energy  from the E-AP, i.e.,  $\ei = \Eh$, but due  to  the wide transmission beam of the E-AP some nodes may receive  more energy  than  they need. Specifically, in some network topologies, the nodes with low energy consumption may receive parts of the energy that is transmitted towards the nodes with higher energy consumption. Accordingly, the stored energy in the batteries of the low energy consumption nodes may grow unbounded. Thus, we let   $\ei \leq \Eh$. That is, the nodes may store only a portion of their received energy. We further assume that $\Td\sum_{l \in \Out_n} p_l(t)$ and $\Eh$ are both upper bounded by $\eimax$, and, consequently, we have $\ei \leq \eimax$ and $ \eo \leq \eimax$.


\item{\textbf{Network Controller}} 
 There exists a  network controller, located at the E-AP, that controls  both the data and the energy links, having access to channel state information (CSI) and the level of the stored data and energy in the queues and batteries  of all  nodes\footnote{We ignore the cost of the nodes sending feedback to the E-AP to inform it about the CSI and stored data/energy. }. The network controller schedules data/energy transmission time sharing by specifying $\Te$ and $\Td$. Moreover, it controls the energy links by specifying the E-AP transmission power $P_{\text{AP}}(t)$  and the beamforming vector $\bm{w}(t)$ and controls the data links  by determining their power vector $\bm{p}(t)$ and   data routing  through specifying $R_{l,s}(t)$. 
\end{description}

%
%
%

Let $E^{\text{opt} }(\bm{\lambda})$ denote the minimum achievable average energy consumption per time-slot of the E-AP over all stabilizing polices. We define  $E^{\text{opt} }(\bm{\lambda})$  as a function of $\bm{\lambda$} to emphasize  the dependency of the energy consumption  on the data arrival rate.  In this way, considering the battery level and the stability constraints,  $E^{\text{opt} }(\bm{\lambda})$  can be found as the solution of 
%
%
%
%
\begin{small}
\begin{mini!}
{\substack{\bm{w}(t),P_{\text{AP}}(t),\bm{p}(t),\\ R_{l,s}(t), \Te,\Td }} { \lim_{T \rightarrow \infty}\frac{1}{T}\sum_{t=0}^{T-1}\mathbb{E}\left\{E_{\text{AP}}(t)\right\}} {\label{prob:mainProbDef}}{E^{\text{opt} }(\bm{\lambda}) = }
\addConstraint{ \limsup_{T \rightarrow \infty}{1 \over T}\sum_{t = 0}^{T-1}\sum_{n,s}\mathbb{E}\left\{U_{n,s}(t)\right\} < \infty, \;\; \forall n,s \label{eq:stableConstraint}}
\addConstraint{\eo \leq B‌‌_{n}(t), \;\;\forall n,t \label{eq:batteryConstraint}}
\addConstraint{\bm{p}(t) \in \Pi,\;    \sum_{s=1}^SR_{l,s}(t) \leq R_l(\bm{p}(t), \bm{g}(t))\label{eq:dataRelatedConstraint}}
\addConstraint{P_{\text{AP}}(t) \in [0,\Pam],\; \lVert \bm{w}(t) \rVert = 1\label{eq:energyRelatedConstraint}}
\addConstraint{\Te+ \Td = \Tf.\label{eq:timingConstraint}}
\end{mini!}
\end{small}
Constraint \eqref{eq:stableConstraint} ensures a finite average backlog and, accordingly, finite average delay \cite[Chapter 2]{Neely2010}. Moreover,  Constraint \eqref{eq:batteryConstraint} is the battery level constraint, which guarantees that the  energy consumption of a node is not greater than the stored energy in the battery of  the node.   Constraints \eqref{eq:dataRelatedConstraint} and \eqref{eq:energyRelatedConstraint} are the restrictions on the data-links and the energy-links parameters, respectively, and \eqref{eq:timingConstraint} is the data/energy transmission time sharing constraint. 

  Problem \eqref{prob:mainProbDef} is a stochastic utility optimization problem.  In every time-slot $t$, the network controller observes the battery levels, the  queue backlogs, the instantaneous CSI as well as the external data arrival and determines the control action. Note that, since the control policy depends on the queue backlogs and the battery levels, the control actions are not necessarily stationary. This problem could be tackled by the  min drift plus penalty (MDPP) algorithm \cite[Chapter 4]{Neely2010}, if the energy consumption of the nodes were not restricted by the battery level.  However, the battery constraint  complicates our problem and makes it  challenging. This is mainly due to the fact that in the battery-operated case, consuming high energy in a specific time-slot may drastically reduce  the battery level  and affect  transmission in the following time-slots. Therefore, having the battery level constraint, policies  with independent decisions at each time-slot are not optimal, which is not acceptable in the  MDPP problem formulation.  
To handle the battery constraint, we define a  Lyapunov function which implicitly takes into account the energy restrictions of the network.  Then, we relax the battery constraint and follow the MDPP approach to design the control policy based on the new Lyapunov function. Finally, we show that the designed control policy conforms to the battery constraint.

\section{The Proposed Control Policy}\label{sec:policy}
In this section, we construct the \alg{}. The general idea behind the \alg{} is to prevent the queue backlog from growing large, while  the energy levels in the batteries of the nodes are kept at an appropriate level in proportion to their stored data backlog level. For this purpose,  we introduce the imbalance indicators $Z_{n}(t),\; \forall n$, as 
\begin{align}\label{eq:Zdef}
Z_{n}(t) \triangleq  \sum_s U_{n,s}(t) - \C B_n(t),
\end{align}
where   $\C\triangleq \frac{2\delta}{1- \frac{1}{\alpha}}$, for some $\alpha >1$. Note that $\C$ represents  {an energy to data conversion factor.} The  value of $Z_{n}(t)$ indicates the data/energy imbalance at  node $n$ in time-slot $t$. For constructing the \alg{}, we follow the  MDPP approach. In summary, we follow the following steps: 
\begin{enumerate}[leftmargin=*]
\item We define the Lyapunov function as
\begin{align}\label{eq:LyapunovFun}
L(t)\triangleq {1\over 2}\lVert \bm{U}(t)\rVert^2 + {1\over 2}\lVert \bm{Z}(t)\rVert^2,
\end{align}
where $\bm{U}(t) \triangleq [U_{n,s}(t), \forall n,s]$ and $\bm{Z}(t) \triangleq [Z_{n}(t), \forall n]$ are the vectors of the backlog level in the data queues and  the imbalance indicators, respectively.  
According to \eqref{eq:LyapunovFun}, the Lypunov function grows if the stored data level in the queues  and/or   the data/energy imbalance  increases. As a result, we intuitively expect  a stabilizing controller to prevent the Lyapunov function from growing large.  
\item We define the Lypunov drift function, which is the expected increment of the Lyapunov function in successive slots, i.e.,
\vspace{-0.4cm}
\begin{align}
\Delta(L(t)) \triangleq \eexp{L(t+1) - L(t)| \bm{U}(t), \bm{B}(t)},
\end{align}
where the expectation is with respect to the randomness in the data and the energy channel gains and the data arrivals.
\item We define the drift-plus-penalty function as 
\begin{align}
\Delta_p(L(t),V) \triangleq  \Delta(L(t)) +V \mathbb{E}\{   E_{\text{AP}}(t)| \bm{U}(t), \bm{B}(t)\}, \label{eq:driftPlusPenalty}
\end{align} 
where $V>0$  is a control parameter.{ We will derive an upper bound for  $\Delta_p(L(t),V)$ in Lemma \ref{lem:upperBound}.}

\item {The \alg{} is designed to approximately  minimize the upper bound obtained in Lemma \ref{lem:upperBound} subject to the instantaneous constraints \eqref{eq:energyRelatedConstraint},        \eqref{eq:dataRelatedConstraint} and \eqref{eq:timingConstraint} but without the battery constraint \eqref{eq:batteryConstraint}.}

\item In Lemma \ref{lem:AlgProp}, we show that the  \alg{} conforms to \eqref{eq:batteryConstraint}. Moreover, in Theorem \ref{th:performanceTh}, we show  that $E_{\text{AP}}(\lambda)$ obtained by our proposed \alg{} is within a bounded distance of $E^{\text{opt}}(\lambda)$, which depending on the considered value of $V$ can be arbitrarily low, and the average stored backlog in the queues satisfies \eqref{eq:stableConstraint}. 

\end{enumerate}

 The details of the analysis are explained as follows. First, Lemma  \ref{lem:upperBound}  derives an upper bound on $\Delta_p(L(t),V)$. 
%

\begin{lemma}\label{lem:upperBound}
 For the  drift-plus-penalty function \eqref{eq:driftPlusPenalty}, we have
\begin{align}
\begin{split}
 &\Delta_p(L(t),V) \leq   \mathcal{B}_0 + F(t) + \sum_{n,s}\eexp{e_{n,s}(t) | \bm{U}(t), \bm{B}(t) }(U_{n,s}(t)+Z_n(t)),
\end{split}\label{eq:dppUpperBound}
\end{align}
where
\begin{align}\label{eq:defFt}
\begin{split}
F(t) &\triangleq \eexp{VE_{\text{AP}}(t)}+ \sum_{n,s} \eexp{ \ui -\uo | \bm{U}(t), \bm{B}(t) } U_{n,s}(t)+ \\
& \quad\sum_{n} \mathbb{E}\bigg\{\sum_s\ui-\sum_s\uo - \C\left[\ei - \eo\right] | \bm{U}(t), \bm{B}(t)\bigg\} Z_{n}(t),
\end{split}
\end{align}
 $\mathcal{B}_0 =   N\times (S\uimax +\C\eimax)^2+ N\times S\times \uimax^2$ and $e_{n,s} = [\uo-U_{n,s}(t)]^+$.
\end{lemma}
\begin{proof}
See Appendix \ref{sec:proofUpperBound}.
\end{proof}
The terms in \eqref{eq:defFt} can be rearranged to  better demonstrate  $F(t)$ as a function of the control variables. Particularly, we write  
\vspace{-0.35cm}
\begin{align}\label{eq:defFtR}
F(t) = \tilde{F}(t)+\eexp{\C(\Eh - \ei) | \bm{U}(t), \bm{B}(t)}Z_n(t),
\end{align}
where
\begin{small}
\begin{align}
\begin{split}
\tilde{F}(t) &= \eexp{ \Td \left(\C\sum_{l=1}^L Z_{\Tnode}(t) p_l(t)-  \sum_{l=1}^L  \sum_s W_{l,s}(t)R_{l,s}(t) \right) \bigg \vert \bm{U}(t), \bm{B}(t)}\\
&+ \eexp{ \Te P_{\text{AP}}(t) \bigg(V- \C\sum_{n=1}^N  |\bm{w}(t)\bm{h}_n^T(t)|^2 Z_n(t)\bigg)  \bigg \vert \bm{U}(t), \bm{B}(t) }+ \sum_{n,s}\lambda_{n,s}\left(U_{n,s}(t)+Z_{n}(t)\right),
\end{split}\label{eq:FtdefR}
\end{align} \end{small}
and 
\vspace{-0.35cm}
\begin{align}
\begin{split}
W_{l,s}(t) = &  Z_{\Tnode}  -Z_{\Rnode} +U_{\Tnode,s}(t)  -   U_{\Rnode,s}(t) .
\end{split}
\label{eq:dataCoeff}
\end{align}
The equality in \eqref{eq:defFtR} can be verified by adding and subtracting $\eexp{\C\Eh  | \bm{U}(t), \bm{B}(t)}Z_n(t)$ to \eqref{eq:defFt} and using the definitions for $\ui$, $\uo$, $\eo$ and $\Eh$.
The \alg{} is designed to approximately minimize  the right hand side of \eqref{eq:dppUpperBound}.  In this way, the control policy under the \alg{} follows the following procedure:

\textbf{Initialization:} Set  $\Ut \triangleq \max\{\eimax( \C +\alpha\delta),\uimax\}$ dummy data units in data queues, i.e., $U_{n,s}(0) = \Ut$ and assume $E_n(0) = 0, \forall n$. 

\textbf{Data Routing in time-slot $t$:}
Calculate the weights $W_{l,s}(t), \; \forall l,s$ according to  \eqref{eq:dataCoeff}.
Let
\begin{align}\label{eq:slDef}
s_l(t) \triangleq \argmax_{s}\left\{W_{l,s}(t)\right\}, \;\forall l,
\end{align}
\vspace{-0.35cm}
and 
\vspace{-0.35cm}
\begin{align}\label{eq:WlDef}
W_l(t) = \max\left\{\max_s \{W_{l,s}(t)\},0\right\},\;\forall l.
\end{align}
 The total capacity of link $l$ is assigned to  stream $s_l(t)$, i.e., 
\begin{align}\label{eq:routing}
\begin{cases}
R_{l,s}(t) = R_l(\bm{p}(t),\bm{g}(t))  & s =  s_l(t),\\
 R_{l,s}(t) = 0& \;\forall s\neq s_l(t).
\end{cases}
\end{align}

\textbf{Data link scheduling in time-slot $t$:} {The transmission power vector $\bm{p}(t)$ is determined by solving}
\begin{align}
 \bm{p}(t) = \argmin_{\tilde{\bm{p}}(t)\in\Pi}{   \sum_{l=1}^L  \bigg[\C Z_{\Tnode}(t) \tilde{ p}_l(t)  -  W_l(t)R_l(\tilde{\bm{p}}(t),\bm{g}(t)) \bigg]}.
\label{eq:MW}
\end{align}

\textbf{Energy link scheduling in time-slot $t$:}
The energy beamforming vector is determined as 
\vspace{-0.15cm}
\begin{align}\label{eq:beamForming}
\bm{w}(t) = \bm{v}^\ast_{max}(t),
\end{align}
\vspace{-0.15cm}
 where  $\bm{v}_{max}(t)$ is the  principal eigenvector of $\bm{H}(t)$ defined as
\begin{align}
\bm{H}(t) \triangleq  \C\sum_{n = 1}^N Z_n(t) \bm{h}_{n}^T(t)\bm{h}^\ast_{n}(t).
\label{eq:sumChan}
\end{align}
The transmission power of the E-AP is determined by
\begin{align}\label{eq:optimalPap}
P_{\text{AP}}(t) = 
\begin{cases}
\Pam, & V<  \C\sum_{n = 1}^N    |\bm{v}^\ast_{max}(t)\bm{h}^T_{n}(t)|^2Z_n(t),  \\
0,& \text{otherwise.}
\end{cases}
\end{align}

\textbf{Data/Energy time sharing in time-slot $t$:}
{Let} 
\begin{align}F_d^\star(t) \triangleq \C\sum_{l=1}^L Z_{\Tnode}(t) p_l(t)-\sum_{l=1}^L  W_{l}(t)R_{l}(\bm{p}(t) ,\bm{g}(t)),\vspace{-0.35cm}\end{align}
 and 
\vspace{-0.35cm}
\begin{align}F_e^\star(t) \triangleq P_{\text{AP}}(t) \bigg(V- \C\sum_{n=1}^N  |\bm{w}(t)\bm{h}_n^T(t)|^2 Z_n(t)\bigg),\end{align}
where $\bm{p}(t)$, $\bm{w}(t)$ and $P_\text{AP}(t)$  are determined in \eqref{eq:MW}, \eqref{eq:beamForming} and \eqref{eq:optimalPap}, respectively. The time sharing rule is  \vspace{-0.25cm}
\begin{align}\label{eq:timeSharing}
\begin{cases}
\Te = \Tf, \Td = 0 & F_e^\star(t) \leq F_d^\star(t),\\
\Te =  0, \Td = \Tf & F_e^\star(t) > F_d^\star(t).
\end{cases}
\end{align}

\textbf{Queues and batteries update in time-slot $t$:}
The portion of the received energy that is stored in  the battery is determined by 
\vspace{-0.15cm}
\begin{align}\label{eq:inEnergy}
\ei = \min \{\Eh, (Z_n(t)- \uimax) /\C\}.
\end{align}
The data queues and batteries are then updated according to \eqref{eq:queueEvolve} and \eqref{eq:batteryEvolve}, respectively.

The \alg{} policy for controlling the data link and the energy link   are summarized in Algorithms \ref{alg:DataLink} and  \ref{alg:EnergyLink}, respectively.

\begin{algorithm}[t]
\caption{\alg{}: Data routing and power scheduling  in time-slot $t$.}\label{alg:DataLink}
\begin{algorithmic}[1]

\State Calculate $W_{l,s}(t),\;\forall l,s,$  according to \eqref{eq:dataCoeff}.
\State  $ s _l(t)\leftarrow \argmax_{\mathrm{s}} {W_{l,\mathrm{s}}(t)}$, $W_l(t) \leftarrow\max\left \{ \max_{\mathrm{s}}\{W_{l,\mathrm{s}}(t)\},0\right\}$.
\State $ R_{l,{s_l(t)}}(t) \leftarrow R_l(\bm{p}(t),\bm{g}(t))$ and  $R_{l,s} = 0\; \forall l, s \neq{s_l(t)} $. \Comment{Data Routing}
\State  
$\bm{p}(t) \leftarrow \argmin_{\tilde{\bm{p}}\in\Pi}   \sum_{l=1}^L  \big[ J_{\Tnode}(t)  \tilde{p}_l -  W_l(t)R_l(\tilde{\bm{p}},\bm{g}(t)) \big].$ \Comment{Data link power scheduling}

\State $F^\star_d (t) \leftarrow  \sum_{l=1}^L  \big[ J_{\Tnode}(t) p_l(t)  -  W_l(t)R_l(\bm{p}(t),\bm{g}(t))\big]$.
\end{algorithmic}
\end{algorithm}

\begin{algorithm}[t]
\caption{\alg{}: Beamforming and energy transmission scheduling  in  time-slot $t$.}\label{alg:EnergyLink}
\begin{algorithmic}[1]
\State Calculate    $\bm{H}(t)$  according to   \eqref{eq:sumChan}.
\State $\bm{v}(t)\leftarrow$ the  principal eigenvectors of $\bm{H}(t)$  and  $\bm{w}(t) \leftarrow \bm{v}^\ast(t)$. \Comment{Beamforming}
\If{$V <  \C\sum_{n = 1}^N    |\bm{w}(t)\bm{h}^T_{n}(t)|^2Z_n(t)$}\Comment{Energy link power scheduling}
\State  $P_{\text{AP}}(t) \leftarrow \Pam$. 
\Else
\State  $P_{\text{AP}}(t) \leftarrow 0$.
\EndIf
\State $F_e^\star(t) \leftarrow VP_{\text{AP}}(t) - \C\sum_{n=1}^N \ei Z_n(t)$.
\end{algorithmic}
\end{algorithm}

\subsection{Discussion on the Proposed Control Policy}
Considering  \eqref{eq:dataCoeff}, the value of $W_{l,s}(t)$ increases if the  data queue in  node $N_t(l)$ is less congested and/or  if there is less data/energy imbalance in   node $N_t(l)$. Hence, according to the routing policy in  \eqref{eq:routing},  we expect that  with  \alg{} the data  will flow towards the nodes with less congested queues and  less data/energy imbalance.  The power allocation policy in  \eqref{eq:MW} devises a compromise between the energy consumption penalty represented by  $\C Z_{\Tnode}(t) p_l(t),\; \forall l$, and the data transmission reward represented by $W_l(t)R_l(\tilde{\bm{p}}(t),\bm{g}(t)),\;\forall l$.  The beamforming policy implies that the energy beam is focused towards the nodes with higher data/energy imbalance and higher energy-link channel gains. Moreover, considering the E-AP transmission power scheduling in  \eqref{eq:optimalPap}, the control  parameter $V$ can be described as  the energy conservativeness indicator of the \alg{}, since  by increasing   $V$ the E-AP transmits  less often. 

The time sharing control parameters, $F_d^\star(t)$ and $ F_e^\star(t)$, can be described as two metrics  representing the gain for data and energy transmission in time-slot $t$, respectively. These two parameters take into account  the level of the stored energy and data in the nodes as well as the CSI  to determine the energy/data transmission gain. 
Finally, the policy for storing the received  energy   in \eqref{eq:inEnergy} implies that with a small value of $Z_n(t)$  a portion of  the received energy may not be stored in the battery, which prevents the battery from being overcharged.   Note that this event mostly occurs when the energy-link channels are not orthogonal, i.e., {$\bm{h}_n(t) \bm{h}_m^H(t) \neq 0, \forall m\neq n$}. Otherwise, according to \eqref{eq:sumChan} and for a small value of  $Z_n(t)$, the beamforming vector will be almost  orthogonal to $\bm{h}_n(t)$. Hence, the received energy in node $n$ will be negligible. 
\vspace{-0.15cm}

\vspace{-0.20cm}

\subsection{Performance Analysis of the Proposed Control Policy}\label{sec:performance}
 In this section, we evaluate the performance of the proposed policy.  In this regard,  Lemma \ref{lem:AlgProp}  introduces some properties that are satisfied with \alg{} in each time-slot. Particularly, we show in Lemma \ref{lem:AlgProp}  that the battery constraint \eqref{eq:batteryConstraint} is satisfied. In Lemma \ref{lem:driftB}, we use the properties of Lemma \ref{lem:AlgProp} to show that \alg{} approximately minimizes the right hand side of \eqref{eq:dppUpperBound}. Finally,  we use the result in Lemma \ref{lem:driftB} to evaluate the energy consumption  with \alg{} and to show that the backlog in the queues satisfy \eqref{eq:stableConstraint}. {First, we present Lemma \ref{lem:AlgProp}.}
\begin{lemma}\label{lem:AlgProp}
With the \alg{}, in each time-slot $t$, we have that
\begin{enumerate}
\item  The imbalance indicator $Z_{n}(t), \; \forall n$, satisfies $Z_{n}(t) \geq \uimax$.
\item  The assigned  rate $R_{l,s}(t)$ is nonzero only if $U_{\Tnode,s}(t) \geq U_0+\uimax, \;\forall l,s$.
\item  The  drained energy $\eo$ is nonzero only if $B_n(t) \geq \eimax, \; \forall n$.
\end{enumerate}
\begin{proof}
See Appendix \ref{sec:ProofAlgProp}.
\end{proof}
\end{lemma}
 The first part of Lemma  \ref{lem:AlgProp} ensures that the stored amounts of energy in the  batteries are bounded in proportion to the stored backlog  in the data queues of the nodes, i.e., $\C B_n(t) \leq \sum_s U_{n,s}(t)-\uimax$. The second part in Lemma \ref{lem:AlgProp} implies that there will be enough data  for transmission, when the outgoing rate from a node is nonzero. Hence, with  \alg{}, we have $e_{n,s}(t) = 0,\; \forall n,s,t$. Moreover, the third part in Lemma \ref{lem:AlgProp} guarantees that when a node transmits data, i.e., $\eo >0$, we have $\eo \leq \eimax \leq B_n(t)$. Hence,   \alg{} conforms to the battery constraint \eqref{eq:batteryConstraint}.

Using the properties in Lemma \ref{lem:AlgProp}, it can be shown that the \alg{}  approximately minimizes the right hand side in \eqref{eq:dppUpperBound}. Specifically, with $e_{n,s}(t)=0$, it suffices to show that \alg{} approximately minimizes $F(t)$ in \eqref{eq:defFt}. For this reason,  let $F^{\text{min}}(t)$  denote the minimum value of $F(t)$  in time-slot $t$ over every alternative  policy, including the policies that violate the battery constraint \eqref{eq:batteryConstraint}, i.e., \vspace{-0.2cm}
\begin{mini}
{\substack{\bm{w}(t),P_{\text{AP}}(t),\bm{p}(t), R_{l,s}(t), \Te,\Td }} {F(t)}{\label{prob:F}}{F^{\text{min}}(t)=}
\addConstraint{\eqref{eq:dataRelatedConstraint},\eqref{eq:energyRelatedConstraint},\eqref{eq:timingConstraint}.}
\end{mini}
  Lemma \ref{lem:driftB} presents the gap between $F(t)$ under \alg{} and  $F^{\text{min}}(t)$.  
\begin{lemma}\label{lem:driftB}
Under  \alg{}, in each time-slot $t$, we have
\begin{align}\label{eq:tPerformance}
F(t) \leq F^{\text{min}}(t) + \mathcal{B}_1,
\end{align}
\vspace{-0.25cm}
where $\mathcal{B}_1 \triangleq \C\eimax(   \C\eimax + \uimax)$.
\begin{proof}
See Appendix \ref{app:proofLem:driftB}.
\vspace{-0.25cm}
\end{proof}
\end{lemma}

Using the bound in \eqref{eq:tPerformance}, and following the Lyapunov optimization Theorem \cite[Theorem 4.2]{Neely2010},  we compare the  energy consumption under  \alg{}  with  $E^{\text{opt}}(\bm{\lambda})$ and bound the time-averaged expected backlog in the queues. Specifically, let $\Lambda$ denote  the set of  data arrival rates that are  inside the capacity region of  the network. Hence, Problem  \eqref{prob:mainProbDef} is feasible  if and only if  $\bm{\lambda} \in \Lambda$. Theorem~\ref{th:performanceTh} characterizes the performance of the \alg{} when $\bm{\lambda}$ is strictly inside  $\Lambda$. Particularly, parts  1 and 2 of Theorem \ref{th:performanceTh} show the optimality of the energy consumption  and the stability of the network under \alg{}, respectively.

\begin{theorem}\label{th:performanceTh}
Suppose that the arrival rates are strictly inside the capacity region, i.e., there is a scalar $\epsilon_{max}$ such that $\forall \epsilon \in (0,\epsilon_{max}]: \bm{\lambda} + \bm{\epsilon} \in \Lambda$, where $\bm{\epsilon}$ is a vector with all entries equal to $\epsilon$. With our proposed \alg{},
\begin{enumerate}
\item The time-averaged expected energy consumption  satisfies
\begin{align}
\limsup_{T\rightarrow \infty}{1\over T} \sum_{t = 0}^{T-1}\mathbb{E}\{ E_{\text{AP}}(t)\} &\leq E^\text{opt}(\bm{\lambda}) + {\mathcal{B}_2 \over V}.
\label{eq:optimalityGap}
\end{align}
\item The queues are stable and the time-averaged expected sum backlog satisfies
\begin{align}
\limsup_{T\rightarrow \infty}{1\over T} \sum_{t = 0}^{T-1} \sum_{n,s}\mathbb{E}\{U_{n,s}(t)\} & \leq {V E^\text{opt}(\bm{\lambda}+\bm{\epsilon}_{max}) + \mathcal{B}_2 \over \epsilon_{max}},
\label{eq:stability}
\end{align}
\end{enumerate}
where $\mathcal{B}_2 \triangleq  \mathcal{B}_0+ \mathcal{B}_1$.
\end{theorem}

\begin{proof}
See Appendix \ref{sec:PrrofThPerf}.
\end{proof}
 The performance bounds in $\eqref{eq:optimalityGap}$ and $\eqref{eq:stability}$ introduce a trade-off between the optimality gap and the average queue backlog that is controlled by $V$.  According to this trade-off, when the average energy consumption is within  $\mathcal{O}({1\over V})$ of the minimum  energy, the  average backlog  could be upper bounded by a term of the order of $\mathcal{O}(V)$.

\color{black}

\section{{Time Evolution of Data Backlog and Battery Level}}\label{sec:behavior}
Theorem \ref{th:performanceTh} bounds the  average backlog in the queues. However, it does not discuss the  behavior  of the backlogs and the battery levels, which are of  importance for the implementation of the policy. In this section, we study the time evolution of the data backlog  and the battery level  under  \alg{}  using the  backlog attraction result in \cite{HuangDelay}. We show that with  \alg{} the data backlog  and the battery level converge to a transformation of the dual optimal solution for the following deterministic problem 
\begin{mini!}
{\substack{\bm{p}(t), \bm{w}(t),C_{(l,s)}(t),P_{\text{AP}}(t), \Te, \Td   }} { V\eexp{E_{\text{AP}}(t)} }{\label{prob:stationaryProbDef}}{V E^{\star} (\bm{\lambda}) = }
\addConstraint{  \eexp{\ui} \leq \eexp{\uo}, \forall n,s} 
\addConstraint{  \eexp{\eo} \leq \eexp{\ei}, \forall n} 
\addConstraint{ \eqref{eq:dataRelatedConstraint}, \eqref{eq:energyRelatedConstraint},  \eqref{eq:timingConstraint}}.
\end{mini!}
The solution to \eqref{prob:stationaryProbDef} is a stationary policy that is only a function of the instantaneous CSI. Hence, we have omitted the time averages.  
Specifically, let $g\big( [\bm{\eta}, \bm{\beta}]\big)$  with $[\bm{\eta}, \bm{\beta}] =[\eta_{n,s} \geq 0 \;\forall (n,s), \beta_n\geq 0 \;\forall n]$ denote the dual function of Problem  \eqref{prob:stationaryProbDef}, that is, 
\begin{small}
\begin{align}\label{eq:gDef}
g\big( [\bm{\eta}, \bm{\beta}]\big) = \inf_{  \substack{ \bm{w}(t),P_{\text{AP}}(t), \bm{p}(t),\\C_{(l,s)}(t),   \Te, \Td    }}  \left[    \eexp{VE_{\text{AP}}(t)}  +\eexp{\sum_{n,s} \eta_{n,s}\left( \ui - \uo\right) + \sum_{n} \beta_n\left( \eo- \ei \right)}     \right],
\end{align}
\end{small}
and let $ [\bm{\eta}^\star, \bm{\beta}^\star] $ denote the optimal solution to the dual problem, i.e., 
\begin{align}
[\bm{\eta}^\star, \bm{\beta}^\star]  = \argmax{g\big( [\bm{\eta}, \bm{\beta}]\big)} \;\;\text{s.t.}\;\;  \bm{\eta}, \bm{\beta} \geq 0.
\end{align}
 Let $[\bm{\nu}^\star, \bm{\zeta}^\star] = [\nu_{n,s}^\star, \;\forall (n,s), \zeta_n^\star, \;\forall n]$ be constructed from $[\bm{\eta}^\star, \bm{\beta}^\star]$ as 
\begin{align}
\nu^\star_{n,s} &= \eta^\star_{n,s} - {\beta^\star_n \over \C},\\
\zeta^\star_n &= {\beta^\star_n\over\C}.
\end{align}
Moreover, let $\bm{ \varepsilon^\star} = [\varepsilon_n^\star, \; \forall n]$ be constructed from $[\bm{\nu}^\star, \bm{\zeta}^\star]$ as 
\begin{align}
\varepsilon_n^\star = {\sum_{s} \nu^\star_{n,s} - \zeta^\star_n \over \C  }.
\end{align}
 Theorem \ref{sec:attrTheorem} presents the main result on the behavior of queue backlogs and battery levels.

\begin{theorem}\label{sec:attrTheorem} Suppose that the dual function \eqref{eq:gDef} satisfies
\begin{align}\label{eq:plolyHedral}
g\big([\bm{\eta}^\star, \bm{\beta}^\star]\big)- g\big([\bm{\eta}, \bm{\beta}]\big)  \geq \mathcal{K} \big\lVert   [\bm{\eta}^\star, \bm{\beta}^\star] -    [\bm{\eta}, \bm{\beta}] \big\rVert,
\end{align}
 for some $\mathcal{K} > 0$. Then, with the \alg{}, there exists constants $D$, $c^\star$ and $\beta^\star$  independent of $V$ such that for every $m \geq 0$ we have
\vspace{-0.25cm}
\begin{align}
\limsup_{T\rightarrow \infty} {1\over T} \sum_{t=0}^{T-1} \Pr \{\exists (n,s): |U_{n,s}(t) - \nu_{n,s}^\star| > D +m \} &\leq {c^\star e^{-\beta^\star m}}\label{eq:ThAttrCu},\\
\limsup_{T \rightarrow \infty}{1\over T} \sum_{t=0}^{T-1}\Pr \left\{ \exists n: |B_n(t)  -  \varepsilon_n^\star |> ((S+1) D + m)/\C  \right\} &\leq {2c^\star e^{-\beta^\star {m\over S+1}}}.\label{eq:ThAttrCe}
\end{align}
\begin{proof}
See Appendix \ref{sec:appAtrractionPoint}.
\end{proof}
\end{theorem}
Theorem \ref{sec:attrTheorem} shows that the probability that the  backlog in data queues and energy level in batteries  deviate from $\bm{\nu}^\star$ and $\bm{\varepsilon}^\star$, respectively,  decreases exponentially as the deviation increases. Note that the assumption in \eqref{eq:plolyHedral} holds when the control parameters are chosen from a finite set \cite{HuangDelay} which is  the case for the digital implementation of the algorithm.

\vspace{-0.25cm}
\section{\alg{} Implementation}\label{sec:Implementation}
In this section, we discuss some  implementation issues related to  \alg{}. Specifically, we study the effect of the limited-capacity data buffers and batteries and the complexity of our proposed policy. \vspace{-0.35cm}
\subsection{Limited Data Buffers and Batteries}\label{sec:LimBufferImp}
A challenge for the implementation of  \alg{} is the limited buffer size and battery capacity of the nodes. Theorem \ref{sec:attrTheorem} ensures that with  sufficiently large batteries and buffers the probability of data or energy {overflow} is small. Hence this limitation does not affect the performance of the policy. Figure \ref{fig:sampPath} depicts a sample time evolution of the backlog and battery level under  \alg{}.    As can be seen, the backlog converges to a constant value. Also,  with a buffer size of  $\SI{2.5}{MBytes}$ and a battery capacity of $\SI{17}{\milli\joule}$ there will be no overflow. However, the behavior of the backlog suggests that, if we tolerate  {dropping a small amount of data} in the initialization phase of the algorithm, we can further reduce the buffer size and battery capacity. Specifically, in the steady state region of Figs. \ref{fig:sampBackLog}  and 
\ref{fig:sampBattery} the backlog and the battery level  fluctuate approximately in  $\SI{0.1}{MBytes}$ and $\SI{2}{\milli\joule}$ intervals, respectively,  which implies that the arrival and departure processes in steady state can be supported by a $\SI{0.1}{MBytes}$  buffer   and a $\SI{2}{\milli\joule}$ battery.  This observation motivates us to modify  \alg{} for limited buffer size and battery capacity implementation. For this reason,  we define  virtual  queues $\tilde{U}_{n,s}(t)$ and $\tilde{E}_n(t)$  associated with each real and finite data queue and battery, respectively.  The virtual queues are not physical queues and are  simple counters inside the controller that are updated as 
\vspace{-0.25cm}
\begin{align}
\tilde{U}_{n,s}(t+1) &= \tilde{U}_{n,s}(t)  + \ui - \uo\label{eq:virtualDataEvolve},\\
\tilde{E}_n(t+1) &= \tilde{E}_{n}(t) +\ei -\eo\label{eq:virtualEnergyEvolve}.
\end{align}
 \alg{} runs based on the values of $\tilde{U}_{n,s}(t)$ and $\tilde{E}_{n}(t)$ instead of the real queues, hence $\ui$, $\uo$, $\ei$ and $\eo$ will have exactly the same value  as  we had in the cases with infinite length real data queue and batteries.  Accordingly, the limited buffer sizes and battery capacities do not affect the decisions of the controller. However, some data units in real buffers may be dropped  due to either buffer overflow or energy outage.
 Let $L_{n,s}(t)$ denote the total number of dropped data units of stream $s$ in node $n$ up to time-slot $t$. Lemma \ref{lem:limitedBattery} bounds the time averaged expected value of $L_{n,s}(t)$ under the modified \alg{} for limited buffers and batteries.

\begin{figure}
\centering
\begin{subfigure}[b]{0.37\textwidth}
\includegraphics[width = \textwidth]{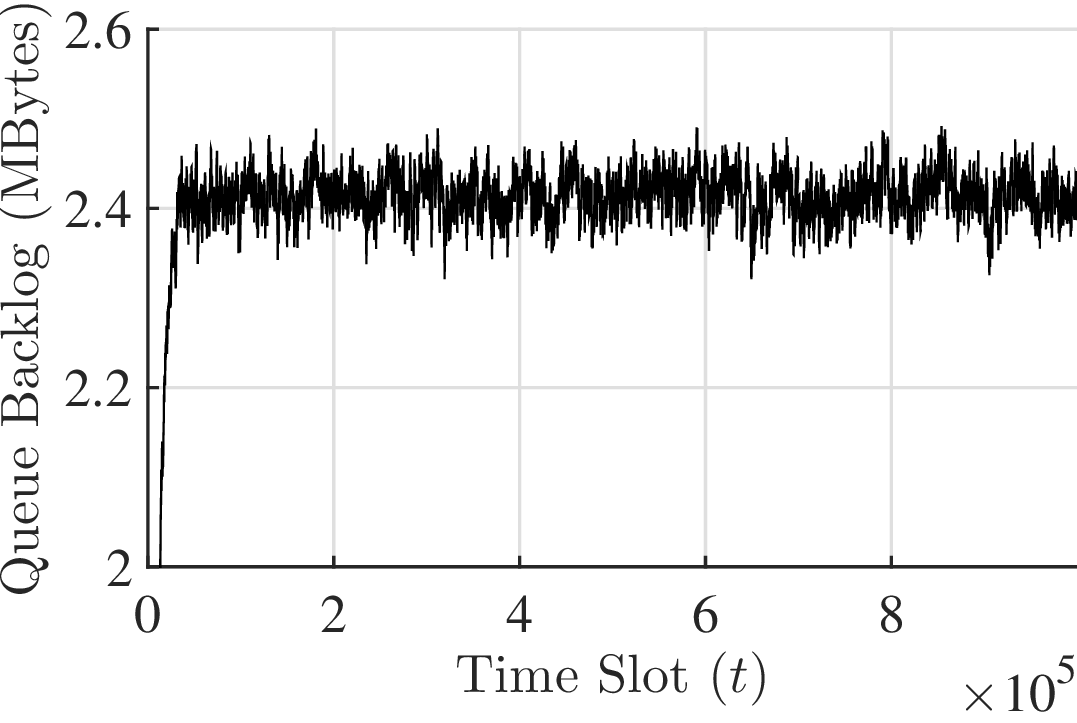}
\caption{Data queue backlog.}
\label{fig:sampBackLog}
\end{subfigure}
~
\begin{subfigure}[b]{0.37\textwidth}
\includegraphics[width =\textwidth]{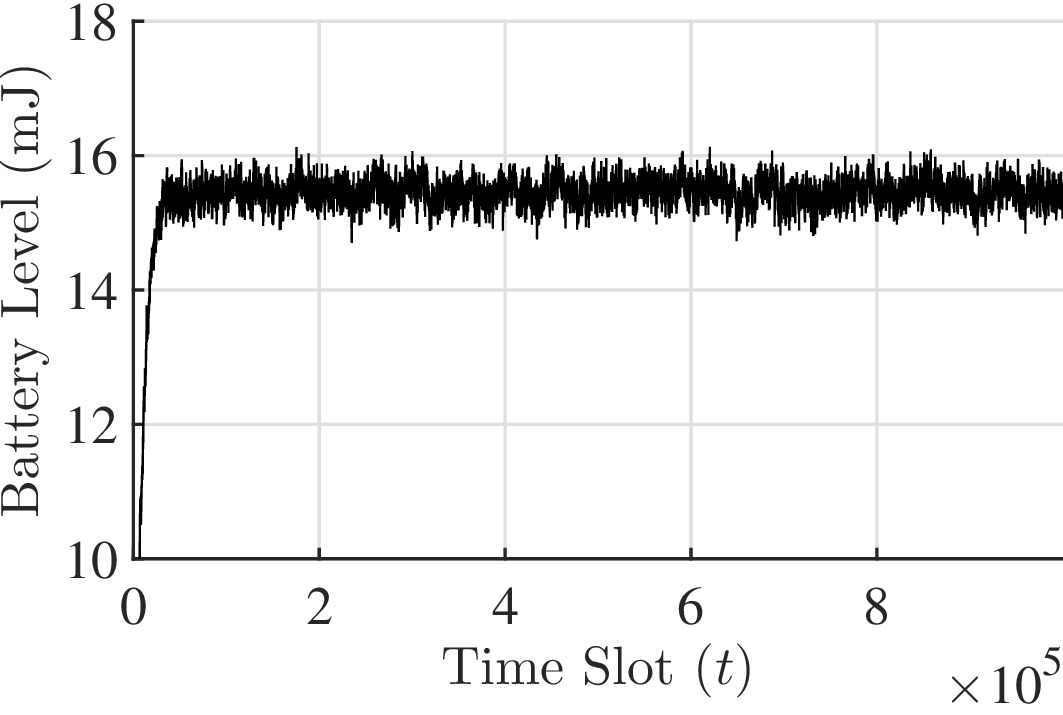}
\caption{Battery level.}
\label{fig:sampBattery}
\end{subfigure}
\vspace{-0.25cm}
\caption{A sample time evolution of the data queue backlog and battery processes. Figures  \ref{fig:sampBackLog} and \ref{fig:sampBattery} correspond to the queue for stream $1$ in node $1$ and the battery in  node $1$ of Fig. \ref{fig:sampNetSim}, respectively. The data arrival rate is $\lambda = 5$ kbps, the Rician $K$-factor $K = 0$ dB and $V = 3\times10^{11}$.  }\vspace{-0.5cm}
\label{fig:sampPath}
\vspace{-0.25cm}
\end{figure}

\begin{lemma}\label{lem:limitedBattery}
 Let $\Uc $ and $\Ec $ denote the size of the data buffers and the capacity of the batteries, respectively. Suppose that $\Uc > 2D +2\uimax$ and $\Ec > {2(S+1)D\over \C} + 2\eimax$. With the modified \alg{},  we have\vspace{-0.2cm}
\begin{align}
\limsup_{T\rightarrow \infty}{1\over T} \eexp{L_{n,s}(T)} \leq  {\uimax c^\star e^{-\beta^\star m_l} } + {2\delta\eimax c^\star e^{-\beta^\star  {m_l\over S+1}} },
\end{align}
where
\vspace{-0.25cm}
\begin{align}\label{eq:m_lDef}
m_l  = \min\bigg\{\Uc/2 - \uimax -D ,\;\C(\Ec/2 - \eimax) -(S+1)D\bigg\}.
\end{align}
\end{lemma}
\begin{proof}
See Appendix \ref{sec:appLimBB}.
\end{proof}
Lemma \ref{lem:limitedBattery} states that the average drop rate decreases exponentially as the buffer size or the battery capacity  increases. 
\vspace{-0.25cm}
\subsection{Complexity of the Proposed Policy }\label{sec:Complexity}
The most computationally expensive part of  \alg{} is solving Problem \eqref{eq:MW}, which is  similar to the well known {max-weight} problem. Under the common interference models, this problem is nonconvex and can be NP-hard \cite{NPhard}. However, many efficient approximate and  distributed solutions are proposed for the max-weight problem  \cite{ 7017589,distributed1, distributed2_j}, that can be extended to solve  \eqref{eq:MW}. As an example,  \cite{7017589} introduces  a distributed iterative algorithm based on the block coordinate descent method for solving a problem similar to \eqref{eq:MW}.

Note that using the same arguments as in \cite{subOptimalScheduling2}, it can be shown that a suboptimal scheduling  in each time-slot may result in satisfactory overall performance. Specifically, instead of \eqref{eq:tPerformance}, if the suboptimal scheduler satisfies
\vspace{-0.25cm}
\begin{align}
F(t) \leq \gamma F^{\text{min}}(t) + \mathcal{B}_3,
\end{align}
in each time-slot $t$ and for some $ \gamma \in [0,1]$ and $\mathcal{B}_3\in \mathbb{R}$,  the time-averaged expected  energy consumption per time-slot will be close to $\gamma E^{\text{opt}}\left({\bm{\lambda}\over \gamma}\right)$.  Accordingly, we may use approximate schedulers with low complexity or reduce the overhead for CSI estimation, while $\gamma$ remains close to  unity and the performance loss is negligible. Below, we study the performance loss due to imperfect CSI  through  simulation.

\section{Simulation Results}\label{sec:sim}

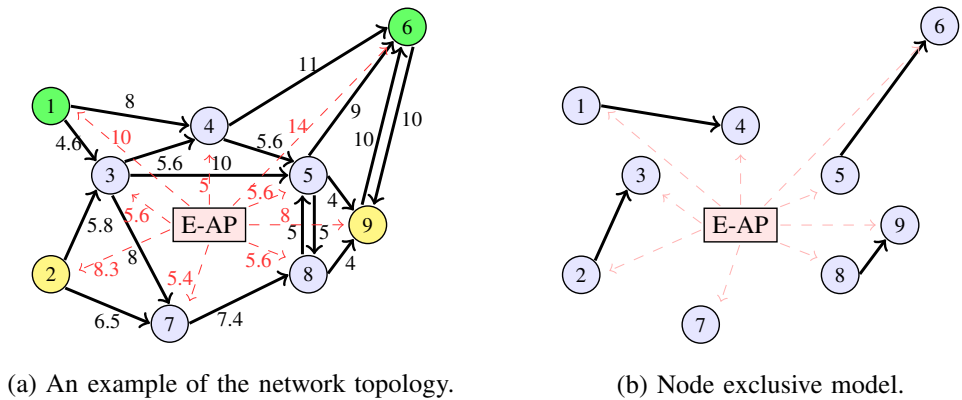
\begin{figure}
\centering
\begin{subfigure}{0.4\textwidth}
\centering
  \begin{tikzpicture}
   [WD/.style={circle,draw,fill=blue!10,scale =0.5},
   HAP/.style={rectangle ,draw,fill=red!10,scale =0.6},
     ST/.style={rectangle ,fill=none,scale =1},labelSt/.style={scale = 0.5},scale=2]

     \node(HAP) at (      0, 0) [HAP]{E-AP};
     \node(WD1) at (    -0.8, 0.6) [WD,fill=green!60] {1};
     \node(WD2) at ( -0.8, -0.25) [WD,fill=yellow!60] {2};
     \node(WD3) at ( -0.5, 0.25) [WD] {3};
     \node(WD4) at (     0, 0.5) [WD] {4};
     \node(WD5) at (  0.5, 0.25) [WD] {5};
     \node(WD6) at (     1, 1) [WD,fill=green!60] {6};
     \node(WD7) at ( -0.2, -0.5) [WD] {7};
     \node(WD8) at (    0.5, -0.25) [WD] {8};
     \node(WD9) at (    0.8,0) [WD,fill=yellow!60] {9};

 	  \draw [->,  shorten >=0.01cm,shorten <=0.01cm , line width = 0.3mm ] (WD1.east) -- (WD4.west) node[midway,above, labelSt] {8}; 		     
 	  \draw [->,  shorten >=0.01cm,shorten <=0.01cm , line width = 0.3mm] (WD1.south east) -- (WD3.north west) node[xshift=-0.12cm, yshift=-0.02cm, midway,  labelSt, rotate = 0 ] {4.6};         
 	  \draw [->,  shorten >=0.01cm,shorten <=0.01cm , line width = 0.3mm] (WD3.north east) -- (WD4.south west) node[xshift=0.1cm, yshift=0.02cm,midway, below,  labelSt, rotate = 0] {5.6}; 
 	  \draw [->,  shorten >=0.01cm,shorten <=0.01cm , line width = 0.3mm] (WD2.north east) -- (WD3.south west) node[xshift=-0.01cm, yshift=0.00cm, midway, right,  labelSt, rotate = 0] {5.8}; 
 	  \draw [->,  shorten >=0.01cm,shorten <=0.01cm , line width = 0.3mm] (WD3.east) -- (WD5.west)node[xshift=0.12cm, yshift=-0.03cm,midway, above,  labelSt, rotate = 0] {10};                    
 	  \draw [->,  shorten >=0.01cm,shorten <=0.01cm , line width = 0.3mm] ([xshift=-0.06cm, yshift=-0.03cm]WD3.south east) -- (WD7.north)node[xshift=0.05cm, yshift=-0.07cm,midway, left,  labelSt, rotate = 0] {8};         
 	  \draw [->,  shorten >=0.01cm,shorten <=0.01cm , line width = 0.3mm] (WD2.south east) -- (WD7.west)node[xshift=-0cm, yshift=-0.0cm, midway, below,  labelSt, rotate = 0] {6.5}; 
 	  \draw [->,  shorten >=0.01cm,shorten <=0.01cm , line width = 0.3mm] (WD7.east) -- (WD8.west)node[xshift=0.12cm, yshift=-0.2cm,midway, left,  labelSt, rotate = 0] {7.4};                    
 	  \draw [->,  shorten >=0.01cm,shorten <=0.01cm , line width = 0.3mm] (WD4.east) -- (WD6.west)node[xshift=0.00cm, yshift=-0.02cm,midway, above,  labelSt, rotate = 0] {11};                    
 	  \draw [->,  shorten >=0.01cm,shorten <=0.01cm , line width = 0.3mm] (WD4.south east) -- (WD5.north west)node[xshift=-0.1cm, yshift=0.1cm,midway, right,  labelSt, rotate = 0] {5.6}; 
 	  \draw [->,  shorten >=0.01cm,shorten <=0.01cm , line width = 0.3mm] (WD5.north) -- (WD6.south west)node[xshift=-0.08cm, yshift=-0.13cm,midway, right,  labelSt, rotate = 0] {9}; 
 	  \draw [->,  shorten >=0.01cm,shorten <=0.01cm , line width = 0.3mm] (WD5.east) -- (WD9.north west)node[xshift=0.05cm, yshift=-0.07cm,midway, left,  labelSt, rotate = 0] {4};           
 	  \draw [->,  shorten >=0.01cm,shorten <=0.01cm , line width = 0.3mm] ( [xshift=0.03cm] WD5.south) -- ([xshift=0.03cm] WD8.north)node[xshift=-0.03cm, yshift=-0.07cm,midway, right,  labelSt, rotate = 0] {5};               
 	  \draw [->,  shorten >=0.01cm,shorten <=0.01cm , line width = 0.3mm]  ([xshift=-0.03cm]WD8.north) -- ([xshift=-0.03cm]WD5.south) node[xshift=0.03cm, yshift=-0.07cm,midway, left,  labelSt, rotate = 0] {5};  
	  \draw [->,  shorten >=0.01cm,shorten <=0.01cm , line width = 0.3mm] (WD8.east) -- (WD9.south west)node[xshift=-0.03cm, yshift=-0.07cm,midway, right, labelSt, rotate = 0] {4};           
 	  \draw [->,  shorten >=0.01cm,shorten <=0.01cm , line width = 0.3mm]  ([xshift=-0.03cm]WD9.north) -- ([xshift=-0.03cm]WD6.south) node[xshift=0.12cm, yshift=+0.07cm,midway, right,  labelSt, rotate = 0] {10};           
 	  \draw [->,  shorten >=0.01cm,shorten <=0.01cm , line width = 0.3mm] ( [xshift=0.03cm]WD6.south) -- ( [xshift=0.03cm]WD9.north)node[xshift=-0.12cm, yshift=-0.13cm,midway, left,  labelSt, rotate = 0] {10};           

         \draw [->,  shorten >=0.1cm,shorten <=0.03cm ,dashed,line width = 0.1mm, color = red!80] ([xshift=0.5ex]HAP.north west) -- (WD1.east)node[xshift=-0.1cm, yshift=0.2cm, midway,  labelSt, rotate = 0] {10}; 					 
         \draw [->,  shorten >=0.1cm,shorten <=0.03cm ,dashed,line width = 0.1mm, color = red!80] ([yshift = -0.1ex]HAP.west) -- (WD2.east)node[xshift=-0.15cm, yshift=-0.18cm, midway,  labelSt, rotate = 0] {8.3}; 					 
         \draw [->,  shorten >=0.1cm,shorten <=0.03cm ,dashed,line width = 0.1mm, color = red!80] (HAP.west) -- (WD3.south east)node[xshift=-0.12cm, yshift=-0.08cm, midway,  labelSt, rotate = 0] {5.6}; 					 
         \draw [->,  shorten >=0.1cm,shorten <=0.03cm ,dashed,line width = 0.1mm, color = red!80] (HAP.north) -- (WD4.south)node[xshift=-0.05cm, yshift=-0.09cm, midway,  labelSt, rotate = 0] {5}; 					 
         \draw [->,  shorten >=0.1cm,shorten <=0.03cm ,dashed,line width = 0.1mm, color = red!80] (HAP.north east) -- (WD5.south west)node[xshift=-0.10cm, yshift=0.06cm, midway,  labelSt, rotate = 0] {5.6}; 					 
 \draw [->,  shorten >=0.1cm,shorten <=0.03cm ,dashed,line width = 0.1mm, color = red!80] ([xshift = -0.4ex]HAP.north east) -- (WD6.south west)node[xshift=-0.15cm, yshift=-0.02cm, midway,  labelSt, rotate = 0] {14}; 					 
\draw [->,  shorten >=0.1cm,shorten <=0.03cm ,dashed,line width = 0.1mm, color = red!80] (HAP.south) -- ([xshift = 0.8ex]WD7.north west)node[xshift=-0.18cm, yshift=-0.02cm, midway,  labelSt, rotate = 0] {5.4}; 					 
\draw [->,  shorten >=0.1cm,shorten <=0.03cm ,dashed,line width = 0.1mm, color = red!80] (HAP.south east) -- (WD8.north west)node[xshift=-0.12cm, yshift=-0.09cm, midway,  labelSt, rotate = 0] {5.6}; 					 
\draw [->,  shorten >=0.05cm,shorten <=0.03cm ,dashed,line width = 0.1mm, color = red!80] (HAP.east) -- (WD9.west)node[xshift=-0.12cm, yshift=0.1cm, midway, labelSt, rotate = 0] {8}; 					 

\end{tikzpicture} 
\caption{An example of the network topology.}
\label{fig:sampNetDist}
\end{subfigure}
~
\begin{subfigure}{0.4\textwidth}
\centering
  \begin{tikzpicture}
   [WD/.style={circle,draw,fill=blue!10,scale =0.5},
   HAP/.style={rectangle ,draw,fill=red!10,scale =0.6},
     ST/.style={rectangle ,fill=none,scale =1},scale=2]

     \node(HAP) at (      0, 0) [HAP]{E-AP};
     \node(WD1) at (    -0.8, 0.6) [WD] {1};
     \node(WD2) at ( -0.8, -0.25) [WD] {2};
     \node(WD3) at ( -0.5, 0.25) [WD] {3};
     \node(WD4) at (     0, 0.5) [WD] {4};
     \node(WD5) at (  0.5, 0.25) [WD] {5};
     \node(WD6) at (     1, 1) [WD] {6};
     \node(WD7) at ( -0.2, -0.5) [WD] {7};
     \node(WD8) at (    0.5, -0.25) [WD] {8};
     \node(WD9) at (    0.8,0) [WD] {9};

 	  \draw [->,  shorten >=0.01cm,shorten <=0.01cm , line width = 0.3mm ] (WD1.east) -- (WD4.west) node[midway,above, scale = 0.4] {}; 		     
 	  \draw [->,  shorten >=0.01cm,shorten <=0.01cm , line width = 0.3mm] (WD2.north east) -- (WD3.south west) node[xshift=-0.01cm, yshift=0.00cm, midway, right,  scale = 0.4, rotate = 0] {}; 
 	  \draw [->,  shorten >=0.01cm,shorten <=0.01cm , line width = 0.3mm] (WD5.north) -- (WD6.south west)node[xshift=-0.08cm, yshift=-0.13cm,midway, right,  scale = 0.4, rotate = 0] {}; 
	  \draw [->,  shorten >=0.01cm,shorten <=0.01cm , line width = 0.3mm] (WD8.east) -- (WD9.south west)node[xshift=-0.03cm, yshift=-0.07cm,midway, right,  scale = 0.4, rotate = 0] {};           

         \draw [->,  shorten >=0.1cm,shorten <=0.03cm ,dashed,line width = 0.1mm, color = red!30] ([xshift=0.5ex]HAP.north west) to (WD1.east); 					 
         \draw [->,  shorten >=0.1cm,shorten <=0.03cm ,dashed,line width = 0.1mm, color = red!30] ([yshift = -0.1ex]HAP.west) to (WD2.east); 					 
         \draw [->,  shorten >=0.1cm,shorten <=0.03cm ,dashed,line width = 0.1mm, color = red!30] (HAP.west) to (WD3.south east); 					 
         \draw [->,  shorten >=0.1cm,shorten <=0.03cm ,dashed,line width = 0.1mm, color = red!30] (HAP.north) to (WD4.south); 					 
         \draw [->,  shorten >=0.1cm,shorten <=0.03cm ,dashed,line width = 0.1mm, color = red!30] (HAP.north east) to (WD5.south west); 					 
 \draw [->,  shorten >=0.1cm,shorten <=0.03cm ,dashed,line width = 0.1mm, color = red!30] ([xshift = -0.4ex]HAP.north east) to (WD6.south west); 					 
\draw [->,  shorten >=0.1cm,shorten <=0.03cm ,dashed,line width = 0.1mm, color = red!30] (HAP.south) to ([xshift = 0.8ex]WD7.north west); 					 
\draw [->,  shorten >=0.1cm,shorten <=0.03cm ,dashed,line width = 0.1mm, color = red!30] (HAP.south east) to (WD8.north west); 					 
\draw [->,  shorten >=0.05cm,shorten <=0.03cm ,dashed,line width = 0.1mm, color = red!30] (HAP.east) to (WD9.west); 					 

\end{tikzpicture} 
\caption{Node  exclusive model.}
\label{fig:sampNetEx}
\end{subfigure}
\vspace{-0.25cm}
\caption{The numbers besides the links in sub-figure (a) show the lengths of the links in meters.  Sub-figure (b) shows a permitted set of active links under the node exclusive mode.  }\label{fig:sampNetSim}
\vspace{-0.65cm}
\end{figure}


In this section, we consider a wireless network consisting of one E-AP and nine wireless nodes, as shown in Fig.  \ref{fig:sampNetDist}. In this network, there are two streams of data, from nodes 1 and 2 to nodes 6 and 9, respectively. We consider  the node exclusive model in which  the data links are orthogonal  but  each node can transmit or receive {only} over a single data link in each time-slot. This model represents Bluetooth networks in which the neighboring nodes transmit over distinct frequencies and each node is equipped with a single half duplex transceiver  \cite{Chaporkar2008}. Accordingly, under the node exclusive model, in each time-slot  only the links that do not share a common node are permitted to be active. Figure \ref{fig:sampNetEx} depicts a sample permitted set of active links under the node exclusive model.

The energy- and data-link CSI follow the Rician fading model \cite{Rician2015}, that is,
\begin{align}\label{eq:hChannel}
\bm{h}_n(t) = \sqrt{\beta_{h_n} K \over K+1} \bar{\bm{h}}_n(t) + \sqrt{ \beta_{h_n} \over K+1}\bm{h}^w_n(t),
\end{align}
and 
\begin{align}\label{eq:gChannel}
{g}_l(t) = \sqrt{\beta_{g_l} K \over K+1} \bar{{g}}_l(t) + \sqrt{\beta_{g_l}  \over K+1}{g}^w_l(t),
\end{align}
where $ \bar{\bm{h}}_n(t)$ and $\bar{{g}}_l(t)$  are the deterministic  component of the channels, and  ${\bm{h}}^w_n(t)$ and ${g}^w_l(t)$ represent the scattered components of the channel. Moreover, $K$ is the Rician $K$-factor  which determines the ratio between the Rician and the scattered components, and $\beta_{g_l}$ and $\beta_{h_n}$  represent the path loss and shadowing effects of the data links and the energy links, respectively. The entries of the energy link scattered component vector ${\bm{h}}^w_n(t)$ and also the data link scattered component ${g}^w_l(t)$ are independent and zero-mean unit variance circularly symmetric complex Gaussian (CSCG) distributed random variables. The deterministic components, $\bar{\bm{h}}_n(t)$ and $\bar{{g}}_l(t)$, are modeled as  \cite[Eq. (2)]{Rician2015}, and the attenuation factors $\beta_{h_n}$ and $\beta_{g_l}$ are calculated at carrier frequency 2.4~GHz.  Furthermore, in all figures we assume $\lambda_1 = \lambda_2 = \lambda$, $\Pm = \SI{1}{\milli\watt}$ and, unless otherwise  mentioned,  we assume $K=20$ dB, $\Pam = 4$ W and $M=20$. We consider the rate-power function in \cite[Eq. (1)]{shortLen} 
\begin{small}
\begin{align}\label{eq:rateShort}
R_l(\bm{p}(t),\bm{g}(t)) = W\left[\log\left(1+{p_l(t)|g_l(t)|^2\over WN_0} \right) - \sqrt{  {1\over \mathcal{L}}\left[1- { \left(1+ {p_l(t)|g_l(t)|^2\over WN_0}\right)^{-2}  }\right]   } Q^{-1}(\rho)\right],
\end{align}
\end{small}
%
where $W $ and $N_0 $ are  the channel bandwidth and the noise power spectral density, respectively.  Moreover, $\mathcal{L}$ is the length of the codewords  and $\rho$ is the maximum block error probability of the decoder. Hence, the second term inside the brackets in \eqref{eq:rateShort} is notable only in the case of  codewords with finite length. Then, letting $\mathcal{L}  \rightarrow \infty$, \eqref{eq:rateShort} is simplified to \eqref{eq:rateLong} for the cases with asymptotically long codewords, \vspace{-0.2cm}
\begin{align}\label{eq:rateLong}
R_l(\bm{p}(t),\bm{g}(t)) = W \log\left(1+{p_l(t)|g_l(t)|^2\over WN_0} \right).
\end{align}
We assume $W=100$ kHz, $N_0 = -135$ dBm/Hz, $\rho = 10^{-10}$  and, except for Fig. \ref{fig:distance} which studies the system performance for short packets, that codewords are sufficiently long such that the second term inside the brackets in \eqref{eq:rateShort} can be neglected.   

 Considering the data arrival rates $\lambda = \{0.5,1.5\}$ kbps and  the number of E-AP transmit antennas  $M = \{20,40\}$, Fig. \ref{fig:powerVsBacklog} demonstrates the  trade-off between  the average energy consumption per time-slot and the average backlog in the queues. The result in Fig. \ref{fig:powerVsBacklog} conforms to the trade-off introduced in Theorem \ref{th:performanceTh}. That is, the average energy consumption is inversely proportional to  the backlog level.  Furthermore, Theorem \ref{th:performanceTh} implies that for sufficiently large $V$ the gap between the  average energy consumption per time-slot  and  $E^{\text{opt}}(\bm{\lambda}) $ is negligible. Hence, the curves in Fig. \ref{fig:powerVsBacklog} converge  to   $E^{\text{opt}}(\bm{\lambda})$. Comparing the curves for $\lambda = 0.5$ kbps and $\lambda = 1.5$ kbps, we observe that the effect of the number of E-AP's transmit antennas  on $E^{\text{opt}}(\bm{\lambda})$  becomes more dominant as the data arrival rate increases.
\begin{figure}
\centering
\includegraphics[width = 0.39\textwidth]{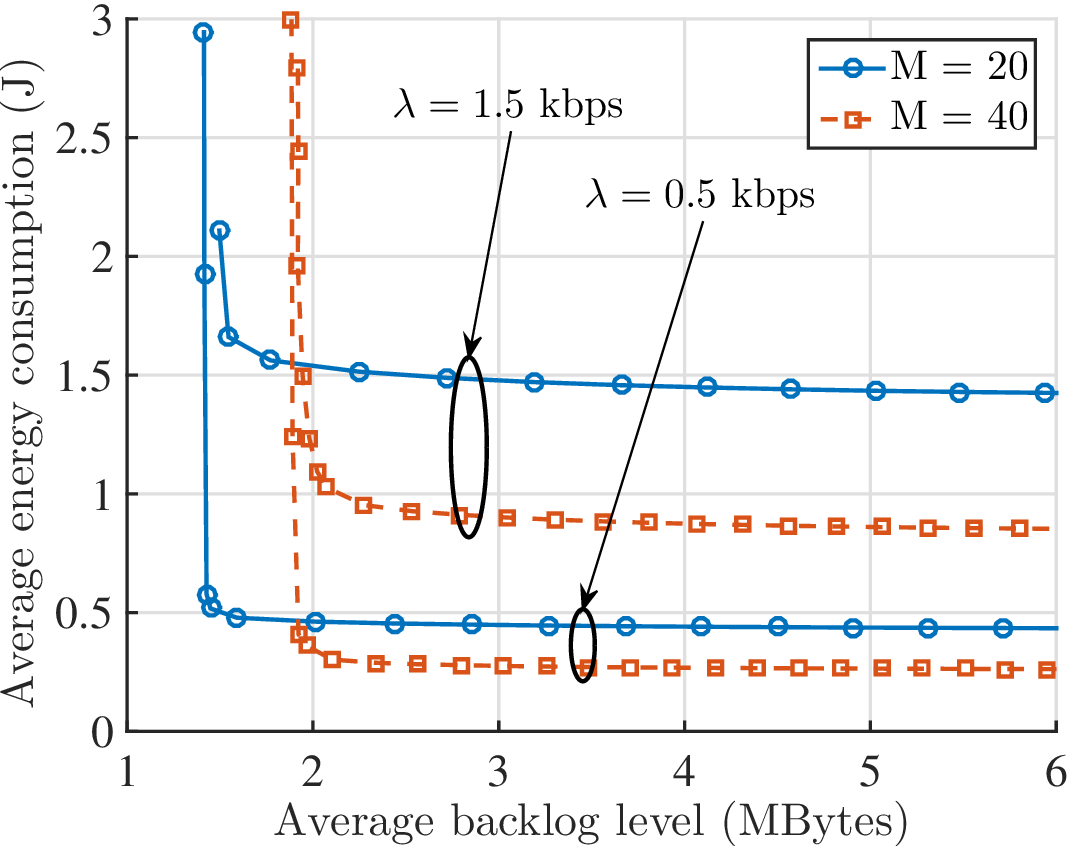}
\vspace{-0.25cm}
\caption{Average energy consumption per time-slot vs the average backlog in data queues. Data arrival rates $\lambda = \{1,3,5\}$ kbps and number of E-AP transmit antennas $M = \{20,40\}$. } \vspace{-0.5cm}
\label{fig:powerVsBacklog}
\vspace{-0.25cm}
\end{figure}

Figures \ref{fig:stream1Flow} and \ref{fig:stream2Flow}  show the average throughput of  streams 1 and  2 over   different data links, respectively. This figure is plotted for $\lambda = 2$ kbps and $V = 10^{11}$. We observe in Fig. \ref{fig:flow} that the data is mostly routed through the shorter links, e.g.,   Fig. \ref{fig:stream1Flow} indicates that  stream 1 reaches  node 4 through  node 3 instead of being directly transmitted. Transmitting over a shorter link reduces the energy consumption of  node 1 that is far from the E-AP and suffers from high energy-link path loss.   Figure \ref{fig:stream2Flow} implies that stream 2 is  routed through two dominant paths. Specifically, the first path includes  nodes 3, 4 and 5, and the second path includes  nodes 7 and 8.  Although the two paths seem to be symmetric according to the topology,  the nodes in the first path are more congested. Hence, the  throughput of stream 2 in the second path is approximately {4.5 times larger than the throughput of} the first path.   Furthermore, the sizes of the nodes in Fig. \ref{fig:flow} represent their average  queue backlog levels, which shows that the average backlog level in the nodes increases when the number of hops between the nodes and the destination of the streams increases. {This is intuitive because under the routing policy in \eqref{eq:routing} and   the link scheduling policy in \eqref{eq:MW} the probability of transmitting stream $s$  over  link $l$ with $U_{N_h(l),s}(t) - U_{N_t(l),s}(t) \leq 0 $ is small.}

\begin{figure}
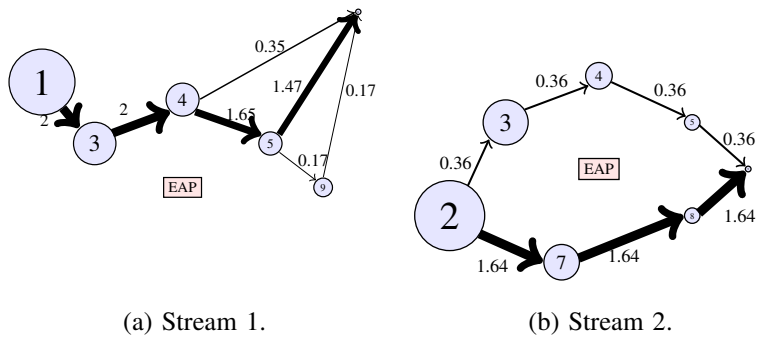

\centering
\begin{subfigure}[b]{0.3\textwidth}
\includegraphics[width = \textwidth]{sampleTopol_Stream1_New.tex}
\caption{Stream 1. }
\label{fig:stream1Flow}
\end{subfigure}
~
\begin{subfigure}[b]{0.3\textwidth}
\includegraphics[width =\textwidth]{sampleTopol_Stream2_New.tex}
\caption{Stream 2. }
\label{fig:stream2Flow}
\end{subfigure}
\vspace{-0.25cm}
\caption{Flow of the data streams in the network, considering $\lambda = 2$ kbps. The numbers on the links are the average throughput of the links in kbps. The thickness of the links and the size of the nodes are proportional to the throughput of the links and average backlog of the queues, respectively.}\vspace{-0.5cm}
\label{fig:flow}
\vspace{-0.25cm}
\end{figure}
 Consider a limited-capacity data buffer and battery implementation of  \alg{} with battery capacity $\mathcal{E}_c = \{0.4,0.8,1.2\}$ mJ and data buffer size $\mathcal{U}_c = \{25, 50 , \ldots, 500\}$ kBytes, Fig. \ref{fig:limBuf} demonstrates the steady state average percentage of the dropped  data with the modified policy of Section \ref{sec:LimBufferImp}.  Here, the results are obtained for $\lambda = 5$ kbps, $K = 0$ dB and $V = 3\times10^{11}$.   As  can be seen in the figure, the percentage of dropped data decreases rapidly as the capacity of the buffer or the battery increases. Specifically,  using batteries with $0.8$ mJ capacity, we observe almost zero drop rate due to energy outage and, consequently, the drop rate becomes independent of the battery capacity for large values of the battery capacity. Moreover,  using data buffers with $200$ kBytes capacity, no data overflow will occur and any further increment of the data buffer size is not necessary.  This result conforms to the result in  Lemma \ref{lem:limitedBattery}, which implies that the average probability of  dropping data units decreases rapidly as the buffer size and the battery capacity increase.

\begin{figure}
\centering
\includegraphics[width =0.39 \textwidth]{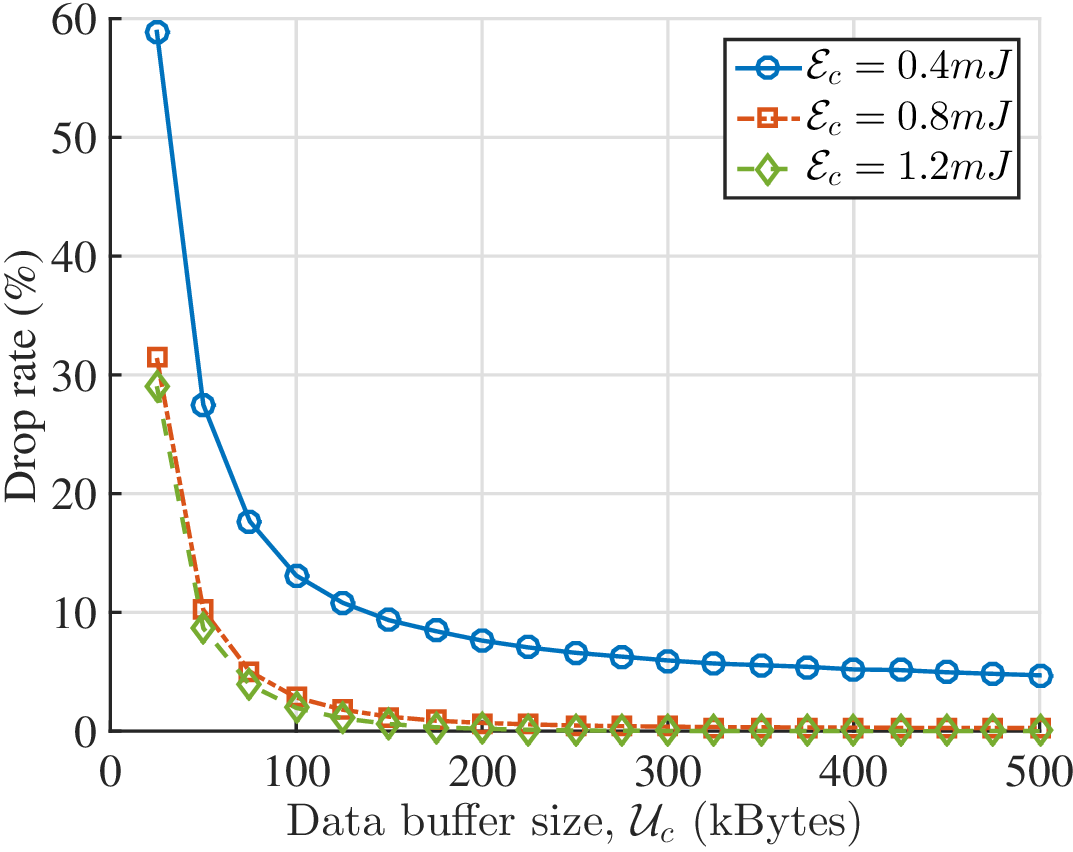}
\vspace{-0.25cm}
\caption{The percentage of dropped bits vs the data buffer capacity and the battery capacity.  Data arrival rate $\lambda = 5$ kbps, $K=0$ dB and $V = 3\times 10^{11}$.  }\vspace{-0.5cm}
\label{fig:limBuf}
\vspace{-0.25cm}
\end{figure}

Considering Rician $K$-factors $K= \{5,10,20\}$ dB and $\lambda =1$ kbps,  Fig. \ref{fig:chanEstErr} studies the effect of the CSI estimation error on the energy consumption. We model  the CSI estimation error as  in  \cite{Rician2015}. In this model, the deterministic component of the channel is assumed to be known and the scattered component is estimated by pilot transmission. Specifically, let  $\hat{\bm{h}}^w_n(t)$ and $\hat{g}^w_l(t)$ denote the estimated scattered component of the energy links and data links, respectively. Moreover,  let  $\tilde{\bm{h}}^w_n(t)\triangleq {\bm{h}^w}_n(t)- \hat{\bm{h}}^w_n(t)$ and $\tilde{g}^w_l(t)\triangleq {g}^w_l(t)- \hat{g}^w_l(t)$ denote the CSI estimation error of the energy links and data links, respectively. The entries of $\tilde{\bm{h}}^w_n(t)$  are i.i.d. zero mean CSCG  random variables with variance $\sigma^2_{h_n}    \triangleq  \left({  \beta_{h_n}   \psi_{p}^h \over \sigma_N(K+1) }+1\right)^{-1}$, and   $\tilde{g}^w_l(t)$ is an i.i.d. zero mean CSCG  random variable with variance $\sigma^2_{g_l} \triangleq \left( {  \beta_{g_l} \psi_p^g\over \sigma_N(K+1) }+1\right)^{-1}$.
Here, $\psi_p^h$ and $\psi_p^g$   are the pilots' energy used for energy link and data link CSI estimation, respectively, and $\sigma_N$  is the variance of the received noise during pilot transmission. In Fig. \ref{fig:chanEstErr}, the energy consumption under  \alg{} is plotted versus  the pilots' energy.  Here, the results are presented for $\sigma_N = -90$ dBm and  sufficiently large values of $V$ such that  the gap between the average energy consumption and $E^{\text{opt}}(\bm{\lambda}) $ is negligible. Moreover, for every value of $K$ three cases are considered, namely, imperfect data-link CSI ($\psi_p^g = \psi_p, \psi_p^h = \infty $), imperfect energy-link CSI ($\psi_p^h = \psi_p, \psi_p^g = \infty $) and imperfect data-link and energy-link CSI ($\psi_p^g = \psi_p^h = \psi_p$), where $\psi_p = \{10^0,10^{0.5}, \ldots, 10^7\}\mu$J.

 An imperfect CSI results in suboptimal scheduling in each time-slot which,  according to the discussions in Section \ref{sec:Complexity}, may still  lead to a satisfactory overall performance. The result in Fig.  \ref{fig:chanEstErr} indicates the excessive energy consumption due to the imperfect CSI-based suboptimal scheduling. With large values of $K$, that is, when the line-of-sight components of the channels are dominant, the effect of imperfect CSI is negligible. Hence, the resources allocated to pilot transmission, i.e., time and energy,  can be saved by avoiding pilot transmission  in every time-slot.  Moreover, as demonstrated in Fig. \ref{fig:chanEstErr}, when   $\psi_p$ exceeds $10^4$,  $10^5$, and $10^6\mu$J for the cases with $K = 5, 10$ and $ 20$ dB, respectively, the energy consumption is almost equal  to that in   the cases with perfect CSI. Hence, any further increment of the pilots' energy has marginal effect on energy consumption. Also, when $\psi_p$ becomes less than $10^2$, $10^3$  and $10^4\mu$J for the cases $K = 5, 10$ and $ 20$ dB, respectively,  the energy consumption becomes independent of the pilots energy. This is because for small values of $\psi_p$ the estimates   $\hat{\bm{h}}^w_n(t)$ and $\hat{g}^w_l(t)$ are almost independent of their exact values. The results for different values of $K$ imply that when   the scattered component  is dominant, i.e., with low values of $K$,  the  energy consumption decreases. This is  intuitive because  \alg{} takes advantage of the  diversity introduced by the scattered component, that is, the nodes avoid transmitting in time-slots with low channel gain and save their energy for possible transmission in subsequent time-slots with higher channel gain.   Also, it should be noted that in practice when the scattered component becomes more dominant the path loss  increases. {Hence,  when reducing the value of $K$, there will be a trade off between the gain introduced by the diversity and the loss due the increased path loss. Here, we have only studied the diversity effect. }

%
%
%
%


\begin{figure}
\centering
\includegraphics[width =0.39 \textwidth]{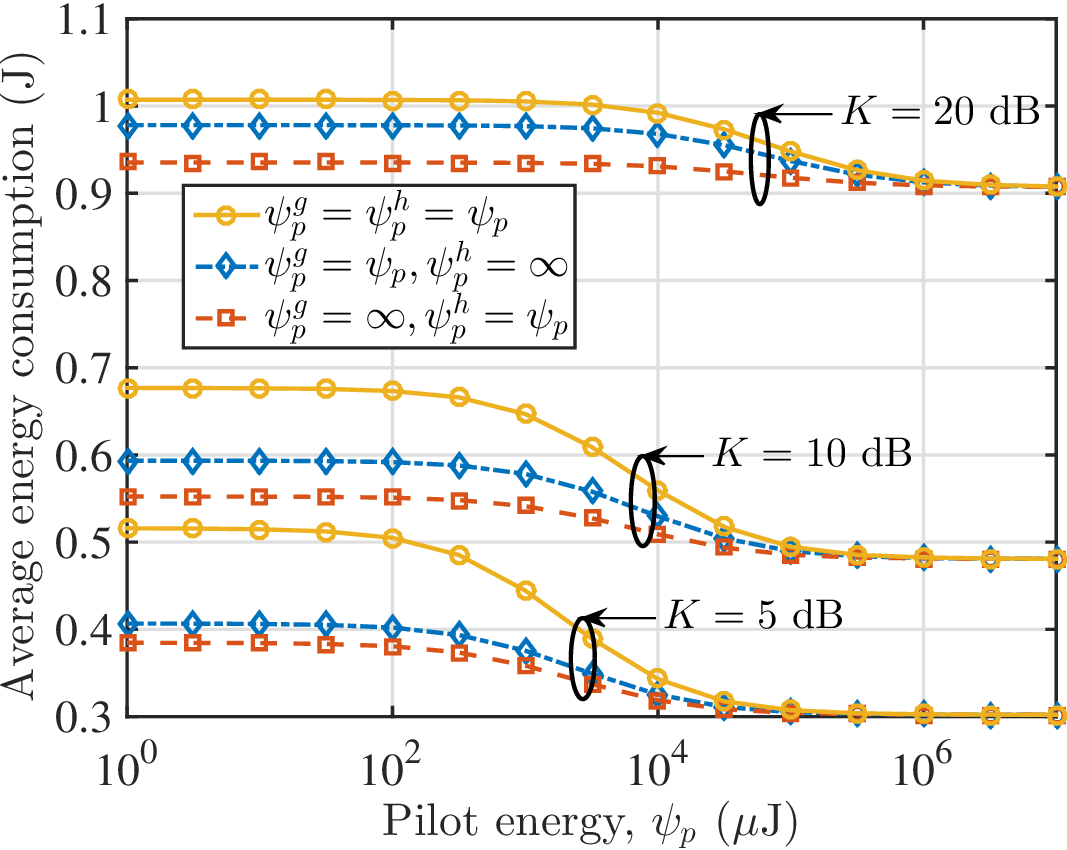}
\vspace{-0.25cm}
\caption{Average energy consumption per time-slot vs the average  backlog in data queues when the energy- and/or data-link CSI is imperfect. The data arrival rate $\lambda= 1 $ kbps.}\vspace{-0.5cm}
\label{fig:chanEstErr}
\end{figure}

Considering the maximum E-AP transmission power  $\Pam = \{3,4,5\}$ W, Fig. \ref{fig:netCap} demonstrates  the   average backlog  in the data queues  versus the  data arrival rate. Theorem \ref{th:performanceTh} states that the average backlog under  \alg{} remains finite if  the input rate is inside the capacity region of the network. Accordingly, Fig. \ref{fig:netCap} shows the maximum value of $\lambda$  that is supported by  \alg{} or every alternative controlling policy. As an example, using Fig. \ref{fig:netCap}, we conclude that for $\Pam = 4$ W no controlling policy can support the streams with arrival rates $\lambda_1 = \lambda_2 \geq 3$ kbps. 

\begin{figure}
\centering
\includegraphics[width =0.39 \textwidth]{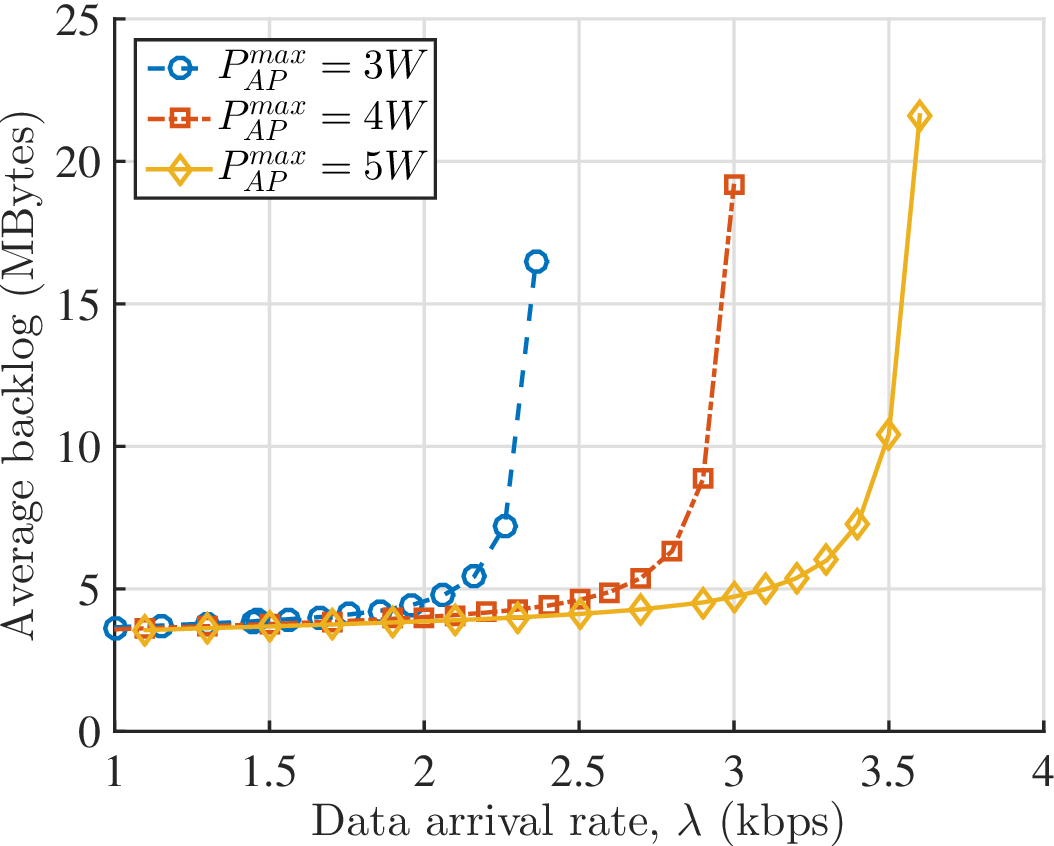}
\vspace{-0.25cm}
\caption{Average backlog in the queues normalized to  data arrival vs the  data arrival rate for maximum transmission power of the E-AP  $P_{\text{AP}} = \{ 3,4,5\}$ W and $V = 10^{11}$. }\vspace{-0.5cm}
\label{fig:netCap}
\vspace{-0.25cm}
\end{figure}

Figure \ref{fig:distance} studies the effect of the nodes' distances and  the finite length codewords on the energy consumption of our proposed policy. For this reason, we consider the  topology in Fig.   \ref{fig:sampNetDist}  and  two scaled versions of this topology,  such that every distance in  Fig.   \ref{fig:sampNetDist}  is scaled by a factor of $1.1$ and $1.2$, respectively.  Figure  \ref{fig:distance} demonstrates the average energy consumption per time-slot under  \alg{} versus the codewords length.   This figure is plotted for sufficiently large values of $V$ such that the gap between the average energy consumption and $E^{\text{opt}}(\bm{\lambda}) $ is negligible. Figure \ref{fig:distance} implies that the average energy consumption is considerably affected by the length of short packets. However, this effect is negligible as the codewords' length increases.  Moreover, the sensitivity of the average energy consumption to the length of short codewords becomes more dominant, when the distances increase. {Also, we observe in  Fig. \ref{fig:distance} that the average energy consumption increases with distance considerably,  because of the high sensitivity of  the path loss to the distance. }

%
%
%

\begin{figure}
\centering
\includegraphics[width =0.4 \textwidth]{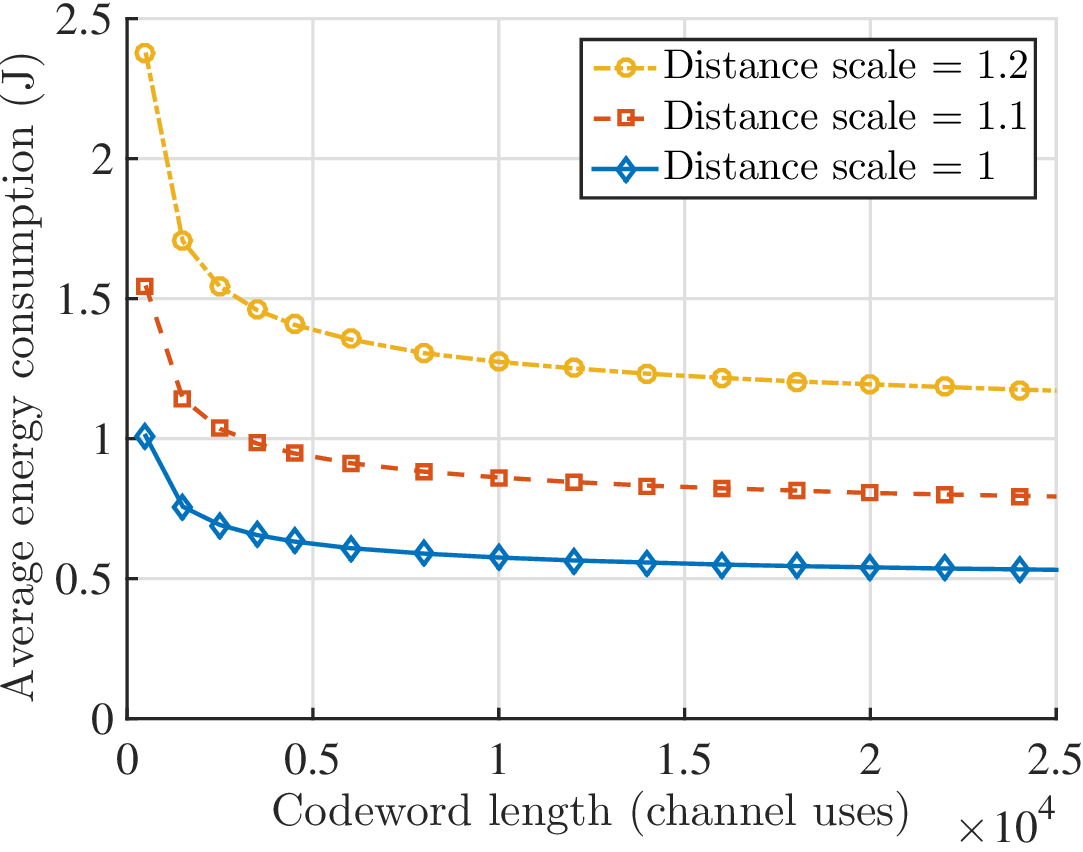}
\vspace{-0.25cm}
\caption{Average energy consumption per time-slot vs the codewords length under scaled distances. 
The distances in  Fig. \ref{fig:sampNetDist}  are scaled by a factor of  $1.1$ and $1.2$.  Data arrival rate $\lambda =0.5$ kbps.}
\label{fig:distance}\vspace{-0.5cm}
\end{figure}

\vspace{-0.25cm}
\section{Conclusion}\label{sec:conclude}
In this paper, we studied  a wirelessly-powered communication network with battery-operated nodes. We proposed  a joint power allocation, data routing, data/energy transmission time sharing and energy beamforming policy to stabilize the network, while minimizing the average energy consumption in the E-AP.  We analyzed the behavior of the backlog in the queues and the stored energy in the batteries. Also, we proposed a modified version of the policy that significantly reduces the data buffer sizes and battery capacities, while dropping only a small portion of the data. As shown, the energy consumption is inversely proportional to the  queue backlogs. {Moreover, with an energy-efficient  routing policy  data is routed through the shorter links and  the nodes that are closer to the E-AP.  Also, as  was observed, the energy consumption  increases in the cases with more dominant  line-of-sight channel component. Finally, the sensitivity of our  performance metrics, i.e., energy consumption and average backlog, to the system parameters such as codeword length and nodes distance  increases as the data arrival rate increases or the capacity of the network reduces. }

\appendices


\vspace{-0.20cm}
\section{Proof of Lemma \ref{lem:upperBound}}\label{sec:proofUpperBound}
 To prove Lemma \ref{lem:upperBound}, we first  introduce  Lemma \ref{lemm:LemmaUpperBoundInnerLemma} which is more general than what is necessary to prove  \eqref{eq:dppUpperBound} in Lemma \ref{lem:upperBound}. However, it will be useful later in the proof of Theorem \ref{sec:attrTheorem}. 
\begin{lemma}\label{lemm:LemmaUpperBoundInnerLemma}
Consider two arbitrary vectors $\bm{\nu}= [\nu_{n,s}\geq 0 , \forall (n,s)]$ and  $\bm{\zeta}= [\zeta_{n,s}\geq0, \forall (n,s)]$. For all time-slots $t$, we have\vspace{-0.2cm}
\begin{small}
\begin{align}\label{eq:LemmaUpperBoundInnerLemma}
\begin{split}
&\lVert \bm{U}(t+1) -\bm{ \nu} \rVert^2 + \lVert \bm{Z}(t+1) -\bm{ \zeta} \rVert^2 - (\lVert \bm{U}(t) -\bm{ \nu} \rVert^2 + \lVert \bm{Z}(t) -\bm{ \zeta} \rVert^2)  \leq \mathcal{B}_0  \\
&\qquad +2 \sum_{n,s}\left(U_{n,s}(t)- \nu_{n,s}\right)\left( \ui- \uo  \right) +2\sum_{n}\left (Z_{n}(t) - \zeta_{n}\right)\bigg(\sum_s \ui-\sum_s\uo \\
&\qquad\qquad-\C\left(\ei -\eo  \right) \bigg)+2\sum_{n,s}e_{n,s}(t)(U_{n,s}(t)+Z_n(t)).
\end{split}
\end{align}
\end{small}
\end{lemma}
\begin{proof}
Considering $U_{n,s}(t+1), \;\forall n,s,$ and \eqref{eq:queueEvolve},  we have
\begin{small}
\begin{align}\label{eq:lemmaUBU}
\begin{split}
&\left(U_{n,s}(t+1) - \nu_{n,s}\right)^2 = \left([U_{n,s}(t) - \uo]^++\ui -\nu_{n,s}\right)^2 =\left(U_{n,s}(t) +e_{n,s}(t) - \uo+\ui -\nu_{n,s}\right)^2 \\
& \qquad\overset{\text{(a)}}{\leq} \left(U_{n,s}(t)- \nu_{n,s}\right)^2+ 2 \left(U_{n,s}(t)- \nu_{n,s}\right)(\ui-\uo)  +2U_{n,s}(t)e_{n,s}(t) +  \left(e_{n,s}(t) + \ui-\uo \right)^2  \\
&\qquad\overset{\text{(b)}}{\leq}  \left(U_{n,s}(t)- \nu_{n,s}\right)^2+ 2 \left(U_{n,s}(t)- \nu_{n,s}\right)(\ui-\uo) + 2U_{n,s}(t)e_{n,s}(t) + \uimax^2,
\end{split}
\end{align}
\end{small}
where (a) holds since the term $-2\nu_{n,s}e_{n,s}(t)\leq 0 $ is removed. The inequality (b)  holds since $ |e_{(n,s)}(t)+\ui-\uo| \leq \uimax$ and, consequently,  $\left(e_{(n,s)}(t)+\ui-\uo\right)^2 \leq \uimax^2$. Furthermore,  considering $Z_n(t+1)$, we have\vspace{-0.15cm}
\begin{small}
\begin{align}\label{eq:lemmaUBZ} 
\begin{split}
(Z_{n}(t+1) - \zeta_{n})^2 &= \left(\sum_sU_{n,s}(t+1)- \C B_n(t+1) - \zeta_{n}\right)^2 \\
&\overset{\text{(a)}}{=}\left (Z_{n}(t)- \zeta_{n}+\sum_s e_{n,s}(t)+\sum_s\left(\ui-\uo\right) -\C\left(\ei-\eo\right)   \right)^2\\
&\overset{\text{(b)}}{\leq} \left(Z_{n}(t) - \zeta_{n}\right)^2+2\left (Z_{n}(t) - \zeta_{n}\right)\left( \sum_s(\ui-\uo) -\C(\ei-\eo)  \right)+ \\
&\qquad 2Z_n(t)\sum_s e_{n,s}(t) + \left( \sum_s\left(e_{n,s}(t) + \ui-\uo\right) -\C\left(\ei-\eo\right) \right)^2 \\
&\overset{\text{(c)}}{\leq}\left(Z_{n}(t) - \zeta_{n}\right)^2+2\left (Z_{n}(t) - \zeta_{n}\right)\left( \sum_s(\ui-\uo) -\C(\ei-\eo)  \right)+\\
&\qquad 2Z_n(t)\sum_s e_{n,s}(t)+ (S\uimax +\C\eimax)^2 , 
\end{split}
\end{align}
\end{small}
where (a) can be verified using \eqref{eq:queueEvolve}, \eqref{eq:batteryEvolve} and \eqref{eq:Zdef}. The inequality (b) comes from $\zeta_n \sum_se_{n,s}(t) \geq 0$, and the inequality (c) holds since $| \sum_s\left(e_{n,s}(t)+\ui-\uo\right) -\C\left(\ei-\eo\right)|\leq S\uimax + \C \eimax$. Taking summation over $n,s$ of both sides in \eqref{eq:lemmaUBU}, we obtain \vspace{-0.15cm}
\begin{small}
\begin{align}\label{eq:lag1}
\begin{split}
||\bm{U}(t+1) - \bm{\nu}||^2 \leq &||\bm{U}(t) - \bm{\nu}||^2 + 2\sum_{n,s}\left(U_{n,s}(t)-\nu_{n,s}(t)\right)(\ui-\uo)+2\sum_{n,s}U_{n,s}(t)e_{n,s}(t)+ \uimax^2.
\end{split}
\end{align}
\end{small}
Moreover, taking summation over $n$ of both sides in \eqref{eq:lemmaUBZ}, we obtain 
\begin{small}
\begin{align}\label{eq:lag2}
\begin{split}
||\bm{Z}(t+1) - \bm{\zeta}||^2 \leq& ||\bm{Z}(t) - \bm{\zeta}||^2 + 2\sum_{n,s}\left(Z_{n,s}(t)-\zeta_{n,s}(t)\right)\bigg(\ui-\uo -\C \left(\ei-\eo\right)\bigg)+\\
&2\sum_{n,s}Z_{n,s}(t)e_{n,s}(t)+ (S\uimax +\C\eimax)^2.
\end{split}
\end{align}
\end{small}
Summing both sides of \eqref{eq:lag1} and \eqref{eq:lag2} and rearranging the terms,  \eqref{eq:LemmaUpperBoundInnerLemma} is proved.
\end{proof}
\begin{proof}[Proof of Lemma \ref{lem:upperBound}]
Noting  that $L(t) ={1\over 2} \lVert \bm{U}(t)\rVert ^2 + {1\over 2}\lVert \bm{Z}(t)\rVert ^2$,  Lemma \ref{lem:upperBound} can be  proved using  \eqref{eq:LemmaUpperBoundInnerLemma} by setting $\bm{\nu}$ and $\bm{\zeta}$ to all zero vectors, adding $VE_{\text{AP}}(t)$ to  both sides  and taking expectation conditioned on $\bm{U}(t)$ and $\bm{Z}(t)$.  
\end{proof}

\vspace{-0.20cm}

\section{ Proof of Lemma \ref{lem:AlgProp}}\label{sec:ProofAlgProp}

\textbf{Proof of the first claim:} We prove the first claim by induction. It is  straightforward to show that $Z_n(t) \geq \uimax$ is satisfied  at $t=0$. Assuming  $Z_n(t) \geq \uimax$ for some $t\geq0$,  we have\vspace{-0.15cm}
\begin{align}\label{eq:inductionProof}
\begin{split}
Z_n(t+1) &= \sum_s U_{n,s}(t+1) -\C B_n(t+1) {=}  Z_n(t) + \sum_s e_{n,s}(t) + \\
&\qquad \sum_s \ui-\sum_s \uo + \C \eo - \C \ei \overset{\text{(a)}}{\geq} Z_n(t)   - \C \ei \overset{\text{(b)}}{\geq} \uimax.
\end{split}
\end{align}
The inequality (a) holds by neglecting the positive terms and noting that
\begin{small}
\begin{align}\label{eq:partAjutify}
 \C \eo-\sum_s \uo & = \Td\left(\C \sum_{l\in\Out_n }p_l(t) - \sum_s \sum_{ l\in\Out_n} R_{l,s}(t)\right)\geq \Td\left(\C \sum_{l\in\Out_n }p_l(t) -  \sum_{ l\in\Out_n} R_{l}(\bm{p}(t),\bm{g}(t))\right)\nonumber \\ 
& \overset{\text{(a)}}{\geq} \Td\left(\C \sum_{l\in\Out_n}p_l(t) -  \delta \sum_{ l\in \Out_n} p_l(t)\right)   \geq 0.
\end{align}
\end{small}
The inequality (a) in  \eqref{eq:partAjutify} comes from  \eqref{eq:rateBound}, and the last inequality holds since $\C\geq \delta$. The inequality (b) in  \eqref{eq:inductionProof} holds because, from \eqref{eq:inEnergy}, we have $ \C \ei \leq Z_n(t) -\uimax$.  Note that we need the assumption $Z_n(t) \geq \uimax$, since otherwise $\ei = \min\{\Eh, (Z_n(t)-\uimax)/\C\}$ becomes negative, which is not feasible.   Equation \eqref{eq:inductionProof} implies $Z_n(t+1) \geq \uimax$. Hence, the first claim is proved.

\textbf{Proof of the second claim:}  Assume\vspace{-0.35cm}
 \begin{align}\label{eq:claim2Ass1}
U_{n,s}(t) \geq\Ut,\; \forall n,s.
\vspace{-0.25cm}
\end{align} 
 In the following, we show that this assumption holds for all time-slots $t$.  Consider  link $\tilde{l}$ and  stream $s_{\tilde{l}}(t)$, which is defined in \eqref{eq:slDef}. We  use contradiction to show that if \vspace{-0.15cm}
\begin{align}\label{eq:claim2Ass2}
U_{N_h(\tilde{l}),s_{\tilde{l}}(t)} (t) \leq\Ut +\uimax,
\end{align}
 no power will be assigned to  link $\tilde{l}$  in the optimal solution of \eqref{eq:MW}, and hence, no data will be transmitted over  link $\tilde{l}$.
 Let $\bm{p}(t)$ denote the optimal solution of \eqref{eq:MW}. Assume   that $p_{\tilde{l}}(t)$ is nonzero and  \eqref{eq:claim2Ass2} holds. Moreover, let $\bar{\bm{p}}(t)$ denote a power vector, such that \vspace{-0.25cm}
\begin{align}\label{eq:pbarDef}
\begin{cases}
\bar{p}_l(t) = p_{l}(t), &  \forall l \neq \tilde{l},\\ 
\bar{p}_{\tilde{l}}(t) = 0.
\end{cases}
\vspace{-0.25cm}
\end{align} 
Then, we have 
\begin{small}
\begin{align}\label{appProp:cont}
\begin{split}
&\sum_l [ \C Z_{\Tnode}(t) p_l(t) - W_l(t) R_l(\bm{p}(t),\bm{g}(t))]- \sum_l [ \C Z_{\Tnode}(t) \bar{p}_l(t) - W_l(t) R_l(\bar{\bm{p}}(t),\bm{g}(t))]  \\
&\quad \overset{\text{(a)}}{=} \C Z_{N_h(\tilde{l})}(t) p_{\tilde{l}}(t)- W_{\tilde{l}} (t)R_{\tilde{l}}(\bm{p}(t),\bm{g}(t))- \sum_{l \neq \tilde{l}} W_l (t) [R_l(\bm{p}(t),\bm{g}(t)) -R_l(\bar{\bm{p}}(t),\bm{g}(t)) ] \\
&\quad \overset{\text{(b)}}{\geq}\C Z_{N_h(\tilde{l})}(t) p_{\tilde{l}}(t)  - W_{\tilde{l}}(t) R_{\tilde{l}}(\bm{p}(t),\bm{g}(t)) \overset{\text{(c)}}{\geq} p_{\tilde{l}}(t)\left[\C Z_{N_h(\tilde{l})}(t) - \delta W_{\tilde{l}}(t)\right ]  \overset{\text{(d)}}{\geq}  p_{\tilde{l}}(t) Z_{N_h(\tilde{l})}(t) (\C- \delta) \geq 0,
\end{split}
\end{align}
\end{small}
where (a), (b) and (c) hold due to the properties of the rate-power function in \eqref{eq:noPowerNoRate}, \eqref{eq:interference} and \eqref{eq:rateBound}, respectively.
To verify that the inequality (d) holds, it suffices to show that $W_{\tilde{l}}(t) \leq Z_{N_h(\tilde{l})}(t)$. Note that, from \eqref{eq:dataCoeff} and \eqref{eq:WlDef},  we have $W_{\tilde{l}}(t) = Z_{N_h(\tilde{l})}(t) - Z_{N_t(\tilde{l})}(t) + U_{N_h(\tilde{l}), s_{\tilde{l}}(t)}(t) - U_{N_t(\tilde{l}), s_{\tilde{l}}(t)}(t)$. The assumptions \eqref{eq:claim2Ass1} and \eqref{eq:claim2Ass2} imply that $U_{N_h(\tilde{l}), s_{\tilde{l}}(t)}(t) - U_{N_t(\tilde{l}), s_{\tilde{l}}(t)}(t) \leq \uimax$. Moreover, from part 1 in Lemma \ref{lem:AlgProp}, we have  $Z_{N_t(\tilde{l})}(t) \geq \uimax$. Accordingly, we obtain $ U_{N_h(\tilde{l}), s_{\tilde{l}}(t)}(t) - U_{N_t(\tilde{l}), s_{\tilde{l}}(t)} (t)- Z_{N_t(\tilde{l})}(t)\leq 0$ and, consequently, $W_{\tilde{l}}(t) \leq Z_{N_h(\tilde{l})}(t)$. The final inequality in \eqref{appProp:cont} contradicts the optimality of $\bm{p}(t)$, and implies that $p_{\tilde{l}}(t) = 0$ when the assumptions  \eqref{eq:claim2Ass1} and \eqref{eq:claim2Ass2} hold.

To complete the proof of the second claim, we need to show that the assumption \eqref{eq:claim2Ass1} holds for all time-slots $t$. For this reason, note that at  $t=0$, \eqref{eq:claim2Ass1} holds due to the initialization step. We show that \eqref{eq:claim2Ass1}  holds for $t\geq0$ by induction. Specifically,  assume  \eqref{eq:claim2Ass1}   holds in time-slot $t$.  For the queues that satisfy   $U_{n,s}(t) \geq \Ut+\uimax$, we have  $U_{n,s}(t+1) \geq \Ut$ since  $\uo \leq \uimax$.  Moreover, consider the pair $(l,s_l(t))$ with  $U_{\Tnode,s_l(t)}(t) < U_0+\uimax$. According to \eqref{appProp:cont}, we have $p_l(t)=0$ and, consequently, $R_l(\bm{p}(t),\bm{g}(t))(t) =0 $, which implies that no data will exit    $U_{\Tnode,s_l(t)}(t)$. Hence,  we have $U_{\Tnode,s_l(t)}(t+1)=U_{\Tnode,s_l(t)}(t)\geq \Ut$. This completes the proof of the second claim.

%
%

\textbf{Proof of the third claim:} We use contradiction to prove the third claim.  Assume
 \begin{align}\label{eq:claim3Ass1}
B_{\tilde{n}}(t) \leq \eimax,
\end{align}
 and consider $\tilde{l} \in \Out_{\tilde{n}}$. Moreover, let $\bm{p}(t)$ be the optimal solution of \eqref{eq:MW}, and assume that  $p_{\tilde{l}}(t)$ is nonzero. We define $\bar{\bm{p}}(t)$ as in \eqref{eq:pbarDef}. Accordingly,  we have 
\begin{small}
\begin{align}\label{appProp:contBat}
\begin{split}
&\sum_l [ \C Z_{\Tnode}(t) p_l(t) - W_l(t) R_l(\bm{p}(t),\bm{g}(t))]- \sum_l [ \C Z_{\Tnode}(t) \bar{p}_l(t) - W_l(t) R_l(\bar{\bm{p}}(t),\bm{g}(t))]  \\
 &\quad \overset{\text{(a)}}{\geq} p_{\tilde{l}}(t)\left[\C Z_{\tilde{n}}(t) - \delta W_{\tilde{l}}(t)\right ]\overset{\text{(b)}}{\geq} p_{\tilde{l}}(t)\left[\C Z_{\tilde{n}}(t) - \delta (Z_{\tilde{n}}(t)+ U_{\tilde{n}, s_{\tilde{l}} (t) }(t)  )  \right ]\\
& \quad = p_{\tilde{l}}(t)Z_{\tilde{n}}(t)\left[\C  - \delta\left (1+ {U_{\tilde{n}, s_{\tilde{l}}(t)}(t)  \over \sum_s U_{\tilde{n}, s}(t) -\C B_{\tilde{n}}(t)} \right)  \right ] \overset{\text{(c)}}{\geq} p_{\tilde{l}}(t)Z_{\tilde{n}}(t)\left[\C  - \delta\left (1+ {U_{\tilde{n}, s_{\tilde{l}}(t)}(t)  \over  U_{\tilde{n}, s_{\tilde{l}}(t)}(t) -\C B_{\tilde{n}}(t)} \right)  \right ]\\
& \quad \overset{\text{(d)}}{\geq} p_{\tilde{l}}(t)Z_{\tilde{n}}(t)\left[\C  - \delta\left (1+ {\C\eimax + \alpha\delta\eimax  \over  \C\eimax + \alpha\delta\eimax-\C \eimax} \right)  \right ]  \overset{\text{(e)}}{\geq} p_{\tilde{l}}(t)Z_{\tilde{n}}(t)\left(\C\left (1-{1\over \alpha}\right)   -2 \delta\right) = 0. \\
\end{split}
\end{align}
\end{small}
Here, (a) follows  the same steps  as (a), (b) and (c) in \eqref{appProp:cont}, and (b) results from removing the negative terms in    $W_{\tilde{l}}(t)$. Moreover, the inequality  (c) comes from  $U_{\tilde{n}, s_{\tilde{l}}(t)}(t) \leq \sum_s U_{\tilde{n}, s}(t)$ and $U_{\tilde{n}, s_{\tilde{l}}(t)}(t) \geq\C B_{\tilde{n}}(t)$. The inequality $U_{\tilde{n}, s_{\tilde{l}}(t)}(t) \geq \C B_{\tilde{n}}(t)$ holds because, from part 2 in Lemma \ref{lem:AlgProp}, we have  $U_{\tilde{n}, s_{\tilde{l}}(t)}(t) \geq \Ut\geq \C\eimax$ and from  assumption \eqref{eq:claim3Ass1} we have $B_n(t) \leq \eimax$. To verify the inequality (d), note that the function ${x \over x-y}$ over the set $ \{(x,y): x\in \mathbb{R}^+,y\in \mathbb{R}^+, x >y\}$ is   decreasing  with respect to $x$ and  increasing  with respect to $y$.  Hence, (d) results from substituting $U_{\tilde{n}, s_{\tilde{l}}(t)}(t)$ with $\C\eimax + \alpha\delta\eimax   \leq \Ut$ and substituting $B_{\tilde{n}}(t)$ with $\C\eimax$. The last  equality (e) can be verified using $\C = \frac{2\delta}{1- \frac{1}{\alpha}}$. The result in \eqref{appProp:contBat} shows that the power vector $\bm{p}(t)$ with nonzero $p_{\tilde{l}}(t)$ cannot be the optimal solution of \eqref{eq:MW}. Hence, the third claim is proved. 
%
%
%
%


\vspace{-0.20cm}

\section{Proof of Lemma \ref{lem:driftB}}\label{app:proofLem:driftB}
Considering \eqref{eq:defFtR},  the intended result in  \eqref{eq:tPerformance} will be proved if we show that\vspace{-0.15cm}
\begin{align}\label{eq:boudForOvercharge}
\C\left(\Eh-\ei\right) Z_n(t) \leq \mathcal{B}_1,
\end{align}
and that  \alg{}  solves\vspace{-0.25cm}
\begin{mini}
{\substack{\bm{w}(t),P_{\text{AP}}(t),\bm{p}(t), R_{l,s}(t),\Te,\Td }} {\tilde{F}(t)}{\label{prob:Ft}}{}
\addConstraint{\eqref{eq:dataRelatedConstraint},\eqref{eq:energyRelatedConstraint},\eqref{eq:timingConstraint},}
\end{mini}
in each time-slot. First, we  show that  \eqref{eq:boudForOvercharge} holds.  Assuming $Z_n(t)  > (\C\eimax +\uimax)$,   we have\vspace{-0.15cm}
$$ (Z_n(t)- \uimax) /\C > \eimax \geq \Eh. \vspace{-0.15cm}$$
Thus, from \eqref{eq:inEnergy}, we have $\ei = \Eh$,  under which   \eqref{eq:boudForOvercharge} holds. Considering $Z_n(t)  \leq (\C\eimax +\uimax)$, we have \vspace{-0.15cm}
$$\C\left(\Eh-\ei\right) Z_n(t) \leq \C\Eh Z_n(t) \leq  \C\eimax(\C\eimax +\uimax) = \mathcal{B}_1.\vspace{-0.15cm}$$
Hence, \eqref{eq:boudForOvercharge}  holds for every $Z_n(t)$. 

We now show that  \alg{} solves \eqref{prob:Ft}. Note that to minimize the expectations in \eqref{eq:FtdefR}, it suffices to minimize the inner terms of the expectation for every given CSI. The structure of $\tilde{F}(t)$ in \eqref{eq:FtdefR} reveals that it can be minimized over the control variables separately. Specifically, irrespective of the time sharing, the power vector $\bm{p}(t)$ and the assigned rates $R_{l,s}(t)$   can be determined by solving \vspace{-0.25cm}
\begin{mini}
{\bm{p}(t), R_{l,s}(t) } {F_d(t)\triangleq \C\sum_{l=1}^L Z_{\Tnode}(t) p_l(t)-\sum_{l=1}^L \sum_{s=1}^S W_{l,s}(t)R_{l,s}(t)}{\label{prob:Fd}}{}
\addConstraint{\eqref{eq:dataRelatedConstraint},}
\end{mini}
and $\bm{w}(t)$ and $P_{\text{AP}}(t)$ can be determined by solving \vspace{-0.15cm}
\begin{mini}
{\bm{w}(t),P_{\text{AP}}(t) } {F_e(t) \triangleq P_{\text{AP}}(t) \bigg(V- \C\sum_{n=1}^N  |\bm{w}(t)\bm{h}_n^T(t)|^2 Z_n(t)\bigg)    }{\label{prob:Fe}}{}
\addConstraint{\eqref{eq:energyRelatedConstraint}.}
\end{mini}
It is straightforward to verify that the routing  and the data-link scheduling policies in \eqref{eq:routing} and \eqref{eq:MW}  solve \eqref{prob:Fd}.  Moreover, using eigenvalue decomposition and noting that  $Z_n(t)$ is nonnegative, it can be verified that the  energy-link scheduling policy in \eqref{eq:beamForming} and \eqref{eq:optimalPap} solves \eqref{prob:Fe}. Accordingly,  $F_d^\star(t)$ and $F_e^\star(t)$ are the optimal values of  problems \eqref{prob:Fd} and \eqref{prob:Fe}, respectively, and  the  optimal  time sharing is determined by the solution of\vspace{-0.15cm}
\begin{mini}
{\Td,\Te} {\Td F_d^\star(t) + \Te F_e^\star(t)}{\label{prob:Ft}}{}
\addConstraint{\eqref{eq:timingConstraint}},
\end{mini}
which is given by \eqref{eq:timeSharing}. This completes the proof of Lemma \ref{lem:driftB}.


\vspace{-0.20cm}
\section{Proof of Theorem \ref{th:performanceTh} }\label{sec:PrrofThPerf}

We follow the Lyapunov optimization method in \cite[Section 4]{Neely2010} to  prove  Theorem \ref{th:performanceTh}.  For this reason, we first define \vspace{-0.35cm}
\begin{mini!}
{\substack{\bm{w}(t),P_{\text{AP}}(t),\bm{p}(t), \\R_l^s(t), \Te,\Td }} { \lim_{T \rightarrow \infty}\frac{1}{T}\sum_{t=0}^{T-1}\mathbb{E}\left\{E_{\text{AP}}(t)\right\}} {\label{prob:mainProbMDef}}{\bar{E}^{\text{opt} }(\bm{\lambda}) \triangleq }
\addConstraint{\limsup_{T\rightarrow \infty} {1\over T}\sum_{t=0}^{T-1} \eexp{\eo} \leq \limsup_{T\rightarrow \infty} {1\over T}   \sum_{t=0}^{T-1} \eexp{\ei} \;\; \forall n\label{eq:avgBatteryConstraint}}
\addConstraint{\eqref{eq:stableConstraint}, \eqref{eq:dataRelatedConstraint}, \eqref{eq:energyRelatedConstraint},\eqref{eq:timingConstraint}.}
\end{mini!}
The defined problem  in \eqref{prob:mainProbMDef} is similar to  \eqref{prob:mainProbDef} except that the battery constraint \eqref{eq:batteryConstraint} is replaced with a less restrictive constraint on the average energy consumption \eqref{eq:avgBatteryConstraint}. Hence, we have $\bar{E}^{\text{opt}} (\bm{\lambda})\leq E^{\text{opt}}(\bm{\lambda})$, with $E^{\text{opt}}(\bm{\lambda})$ being the solution to Problem \eqref{prob:mainProbDef}. The new Problem \eqref{prob:mainProbMDef} follows the same   framework as introduced  in \cite[Eq. (4.31)-(4.35)]{Neely2010}.  According to  \cite[Theorem 4.5]{Neely2010}, for every  $\sigma >0$ and for data arrival rate $\bm{\lambda}+ \bm{\epsilon}$ with $\epsilon< {\epsilon}_{\max}$, there is a stationary policy that  is only a  function of  the instantaneous CSI and the data arrivals, and  satisfies \eqref{eq:dataRelatedConstraint}, \eqref{eq:energyRelatedConstraint}, \eqref{eq:timingConstraint}. Under the stationary policy in  \cite[Theorem 4.5]{Neely2010}    in each time-slot $t$ we have
\begin{align}
\begin{split}
\eexp{E_{\text{AP}}(t)} &\leq \bar{E}^\text{opt}(\bm{\lambda}+\bm{\epsilon}) + \sigma,\\
\lambda_{n,s}+\epsilon + \eexp{\Td \sum_{l\in\In_n} R_{l,s}(t)} &\leq \eexp{\Td\sum_{l\in\Out_n} R_{l,s}(t)} + \sigma\;\; \forall n,s,\\
\eexp{ \eo  }& \leq \eexp{ \ei  }+ \sigma \;\; \forall n.
\end{split}
\label{eq:rndPolicyProp}
\end{align}

Note that  \cite[Theorem 4.5]{Neely2010} only states the  existence  of a  stationary policy with properties in \eqref{eq:rndPolicyProp} while it does not derive such policy. However, using the properties in \eqref{eq:rndPolicyProp}  and following  the steps in \cite[Theorem 4.2]{Neely2010} and \cite[Theorem 4.8]{Neely2010},  Theorem  \ref{th:performanceTh} can be proved. We  sketch the outline of the proof for the readers convenience. 
From Lemma \ref{lem:upperBound} and the fact that $e_{n,s}(t) = 0$ under  \alg{}, we have \vspace{-0.15cm}
\begin{align}\label{eq:keyMDPP}
\begin{split}
\Delta(L(t)) &+V \eexp{   E_{\text{AP}}(t)| \bm{U}(t), \bm{B}(t)}  \leq \mathcal{B}_0 + F(t).
\end{split}
\vspace{-0.15cm}
\end{align}
Let $F^\text{stat}(t)$ denote the value of $F(t)$ under the stationary policy satisfying  \eqref{eq:rndPolicyProp}. Using  \eqref{eq:rndPolicyProp} in \eqref{eq:defFt}  with $\sigma \rightarrow 0$ and noting that $\bar{E}^{\text{opt}} (\bm{\lambda}+\bm{\epsilon})\leq E^{\text{opt}}(\bm{\lambda}+\bm{\epsilon})$, it can be verified that\vspace{-0.15cm}
\begin{align}\label{eq:inequalityStat}
F^\text{stat}(t) \leq V E^\text{opt}(\bm{\lambda}+\bm{\epsilon}) - \epsilon \sum_{n,s}U_{n,s}(t).
\end{align}
From \eqref{eq:tPerformance}, we have $F(t) \leq \mathcal{B}_1 + F^{\text{min}}(t) \leq  \mathcal{B}_1 + F^{\text{stat}}(t)$, which together with  \eqref{eq:keyMDPP} and  \eqref{eq:inequalityStat} imply\vspace{-0.15cm}
\begin{align}\label{eq:keyMDPP_2}
\Delta(L(t)) &+V \eexp{   E_{\text{AP}}(t)| \bm{U}(t), \bm{B}(t)} \leq \mathcal{B}_2 + V E^\text{opt}(\bm{\lambda}+\bm{\epsilon}) - \epsilon \sum_{n,s}U_{n,s}(t),
\end{align}
with $\mathcal{B}_2 =  \mathcal{B}_0+ \mathcal{B}_1$. Taking expectation with respect to $ \bm{U}(t)$ and $\bm{E}(t)$ from both sides in \eqref{eq:keyMDPP_2} results in\vspace{-0.15cm}
\begin{align}\label{eq:keyMDPP_3}
\eexp{L(t+1)}& - \eexp{L(t)}  + V\eexp{E_{\text{AP}}(t)}  \leq \mathcal{B}_2+VE^\text{opt}(\bm{\lambda}+\bm{\epsilon}) - \epsilon \sum_{n,s}  \eexp{U_{n,s}(t)}.
\end{align}
Summing both sides of \eqref{eq:keyMDPP_3} over $t = 0,\ldots, T-1$ yields\vspace{-0.15cm}
\begin{align}\label{eq:keyMDPP_4}
&\eexp{L(T-1)} - \eexp{L(0)} + V\sum_{t=0}^{T-1}\eexp{E_{\text{AP}}(t)} \leq \big(\mathcal{B}_2+ VE^\text{opt}(\bm{\lambda}+\bm{\epsilon})\big)T - \epsilon \sum_{t=0}^{T-1} \sum_{n,s} \eexp{U_{n,s}(t)}.
\end{align}
By rearranging the terms in \eqref{eq:keyMDPP_4} and dropping the negative terms whenever appropriate, we would have \vspace{-0.25cm}
\begin{align}
{1 \over T}\sum_{t=0}^{T-1}\eexp{E_{\text{AP}}(t)} &\leq {\mathcal{B}_2\over V} +E^\text{opt}(\bm{\lambda}+\bm{\epsilon}) + {\eexp{L(0)}\over T},\label{eq:opt1}\\
{1 \over T} \sum_{t=0}^{T-1} \sum_{n,s} \eexp{U_{n,s}(t)} &\leq {\mathcal{B}_2+VE^\text{opt}(\bm{\lambda}+\bm{\epsilon}) \over \epsilon} +  {\eexp{L(0)}\over T}.\label{eq:stab1}
\end{align}
{The bounds in \eqref{eq:opt1} and \eqref{eq:stab1} can be separately optimized over values of $ \epsilon \in (0, \epsilon_{max}] $. Letting $\epsilon \rightarrow 0$ in \eqref{eq:opt1} and  $\epsilon = \epsilon_{max}$  in \eqref{eq:stab1} and  taking limits as $T\rightarrow \infty$ concludes the  claims of Theorem \ref{th:performanceTh}.}

\vspace{-0.25cm}

\section{Proof of Theorem  \ref{sec:attrTheorem}}\label{sec:appAtrractionPoint}
Let $D(t)$ denote the distance between $[\bm{U}(t),\bm{Z}(t)]$ and $[\bm{\nu}^\star, \bm{\zeta}^\star]$, i.e.,
$
D(t)  = \lVert [\bm{U}(t), \bm{Z}(t)] -[\bm{\nu}^\star, \bm{\zeta}^\star] \rVert.
$
To prove  Theorem  \ref{sec:attrTheorem}, we need Lemma  \ref{eq:lemmaAttr}, which bounds the variation of  $D(t)$ in successive slots.
\begin{lemma}\label{eq:lemmaAttr}
With  \alg{}, there is a constant $\tilde{\mathcal{K}} >0$ such that for time-slot $t$ we have\vspace{-0.15cm}
\begin{align}\begin{split}
\eexp{D^2(t+1) - D^2(t) |   \bm{U}(t),\bm{Z}(t)} \leq \mathcal{B}_2 - 2\tilde{\mathcal{K}}D(t).
\end{split}\label{eq:lemmaAttrBound}
\end{align}

\begin{proof}
Let $Y(t) \triangleq \eexp{D^2(t+1) - D^2(t) |   \bm{U}(t),\bm{Z}(t)} $ from \eqref{eq:LemmaUpperBoundInnerLemma} with $\bm{\zeta} = \bm{\zeta}^\star$, $\bm{\nu} = \bm{\nu}^\star$ and $e_{n,s}(t)=0$, we have\vspace{-0.15cm}
\begin{small}
\begin{align}\label{eq:lemmaAttrE1}\begin{split}
Y(t)& \leq \mathcal{B}_0  +2\sum_{n,s}\mathbb{E} \bigg\{ \left(  U_{n,s}(t) -\nu_{n,s}^\star\right)\left(\ui-\uo\right)  \bigg\} +\\
&\quad 2 \sum_n\mathbb{E} \left\{ \left(Z_{n}(t) - \zeta_{n}^\star \right)\left(\sum_s\left(\ui - \uo\right) - \C(\ei-\eo)\right)\right\}.
\end{split}
\end{align}
\end{small}
Adding and subtracting $2\eexp{Vp_\text{AP}(t)}$ to the left hand side  in \eqref{eq:lemmaAttrE1} and rearranging the terms, we obtain\vspace{-0.15cm}
\begin{small}
\begin{align}\label{eq:lemmaAttrE2}
\begin{split}
Y(t) &\leq \mathcal{B}_0 + 2\mathbb{E}\bigg\{VP_\text{AP}(t)+  \sum_{n,s}U_{n,s}(t)\left(\ui -\uo\right) +\\
&\quad \qquad \sum_{n} Z_{n}(t)\left(\sum_s\left(\ui-\uo\right) - \C\left(\ei - \eo\right)\right)\bigg\}-\\
&\quad2\mathbb{E}\bigg\{  VP_\text{AP}(t)+   \sum_{n,s} \nu_{n,s}^\star\left(\ui-\uo\right) +\sum_{n} \zeta_{n}^\star\left(\sum_s\left(\ui-\uo\right) - \C\left(\ei-\eo\right)\right)\bigg\}.
\end{split}
\end{align}
\end{small}
Let $\bm{N}(t)$ and $\bm{B}(t)$ denote the vectors $[N_{n,s}(t) \;\forall(n,s)]$ and $[B_{n}(t) \;\forall n]$ where\vspace{-0.15cm}
\begin{align}\label{eq:ThAttrPNB}
\begin{split}
N_{n,s}(t) &\triangleq U_{n,s}(t) + Z_{n}(t),\\
B_n(t) &\triangleq \C  Z_{n}(t).
\end{split}
\end{align} 
Using \eqref{eq:ThAttrPNB},  $\eta_{n,s}^\star = \nu_{n,s}^\star + \zeta_{n}^\star$ and $\beta_n^\star = \C{\zeta_{n}^\star}$, \eqref{eq:lemmaAttrE2} can be rewritten as \vspace{-0.15cm}\begin{small}
\begin{align}\label{eq:lemmaAttrE3}
\begin{split}
Y(t) &\leq \mathcal{B}_0 +    2\eexp{VP_\text{AP}(t)+  \sum_{n,s}N_{n,s}(t)\left(\ui-\uo\right) - \sum_{n} B_n(t)\left( \ei-\eo\right)}\\
&\qquad-2\eexp{  VP_\text{AP}(t)+   \sum_{n,s} \eta_{n,s}^\star \left(\ui-\uo\right) -\sum_{(n)} \beta_n^\star \left(\ei-\eo\right)}.
\end{split}
\end{align}
\end{small}
From \eqref{eq:tPerformance},  it can be verified that  \alg{} approximately minimizes the first expectation in the right hand side of  \eqref{eq:lemmaAttrE2} and accordingly  \eqref{eq:lemmaAttrE3}. Hence, we can write\vspace{-0.15cm}
\begin{small}
\begin{align}\label{eq:lemmaAttrE4}
\begin{split}
Y(t) &\leq \mathcal{B}_0+\mathcal{B}_1 +    2 \min_{     \substack{ \bm{w},P_{\text{AP}}, \bm{p},\\C_{(l,s)},   \Te, \Td }    }  \eexp{VP_\text{AP}+  \sum_{n,s}N_{n,s}(t)\left(\ui-\uo \right)- \sum_{n} B_n(t) \left(\ei - \eo\right)}\\
&\qquad-2\min_{     \substack{ \bm{w},P_{\text{AP}}, \bm{p},\\C_{(l,s)},   \Te, \Td }      }\eexp{  VP_\text{AP}+   \sum_{n,s} \eta_{n,s}^\star(\ui-\uo) -\sum_{(n)} \beta_n^\star (\ei-\eo)}\\
& = \mathcal{B}_2 + 2 g\left([\bm{N}(t), \bm{B}(t)] \right) - 2 g\left( [\bm{\eta}^\star, \bm{\beta}^\star]   \right)\leq \mathcal{B}_2 - 2\mathcal{K}\lVert [\bm{N}(t), \bm{B}(t)] -[\bm{\eta}^\star, \bm{\beta}^\star] \rVert,
\end{split}
\end{align}
\end{small}
where the first equality holds according to \eqref{eq:gDef}, and the second inequality results from \eqref{eq:plolyHedral}.  Note that $[\bm{N}(t), \bm{B}(t)]$ is constructed from $[\bm{U}(t), \bm{Z}(t)]$ by a linear one-to-one transform. Hence, there is a constant $\tilde{\mathcal{K}} >0$ such that \vspace{-0.15cm}
\begin{align}\label{eq:ThAttrDistanceBound}
\lVert [\bm{U}(t), \bm{Z}(t)] -[\bm{\nu}^\star, \bm{\zeta}^\star] \lVert \leq {\mathcal{K}\over \tilde{\mathcal{K}}}\lVert [\bm{N}(t), \bm{B}(t)] -[\bm{\eta}^\star, \bm{\beta}^\star]\lVert.
\end{align}
Using \eqref{eq:ThAttrDistanceBound} in \eqref{eq:lemmaAttrE4} completes the proof of Lemma \ref{eq:lemmaAttr}.
\end{proof}
\end{lemma}


Note that \eqref{eq:lemmaAttrBound} implies that the distance between $[\bm{U}(t), \bm{Z}(t)]$ and  $[\bm{\nu}^\star,\bm{\zeta}^\star]$ does not grow large, that is,  the expected gradient of their distance $D(t)$ is negative when $D(t)$ is greater than $B^2 /2\tilde{\mathcal{K}}$.  Lemma \ref{eq:lemmaAttr} together with the exponential Lyapunov drift  analysis in  \cite[Theorem 1]{HuangDelay} imply that there are constants $D$, $c^\star$ and $\beta^\star$ such that\vspace{-0.15cm}
\begin{align}
&\limsup_{T\rightarrow \infty} {1\over T} \sum_{t=0}^{T-1} \Pr \{\exists (n,s): |U_{n,s}(t) - \nu_{n,s}^\star| > D+m \} \leq {c^\star e^{-\beta^\star m}},\label{eq:ThAttrCuP}\\
&\limsup_{T\rightarrow \infty} {1\over T} \sum_{t=0}^{T-1} \Pr \{\exists (n): |Z_{n}(t) - \zeta_{n}^\star| > D+m \} \leq {c^\star e^{-\beta^\star m}}.\label{eq:ThAttrCz}
\end{align}
 The proof of \eqref{eq:ThAttrCuP}  and \eqref{eq:ThAttrCz} is similar to  \cite[Theorem 1]{HuangDelay} and is omitted  for brevity. Here, we use \eqref{eq:ThAttrCuP} and \eqref{eq:ThAttrCz} to  prove \eqref{eq:ThAttrCe}. According to the definition of   $Z_{n,s}(t)$ in \eqref{eq:Zdef}, we have
\begin{align}\label{eq:LemmaAttrEZ}
B_n(t) = {\sum_s U_{n,s}(t) - Z_{n}(t) \over \C  }.\vspace{-0.15cm}
\end{align}
 Moreover, we have \vspace{-0.15cm}
\begin{align}\label{eq:LemmaAttrEpZet}
\varepsilon_n^\star =  {\sum_{s} \nu^\star_{n,s} - \zeta^\star_n \over \C  }. 
\end{align}
According to \eqref{eq:LemmaAttrEZ} and \eqref{eq:LemmaAttrEpZet}, $B_n(t)$ is within $((S+1)D+m)/ \C $ distance of $\varepsilon_n^\star$, whenever $U_{n,s}(t), \forall s,$ and  $Z_{n}(t)$ are  within $D+{m\over (S+1)} $ distance of $\nu_{n,s}^\star$  and  $\zeta_{n}^\star$, respectively. Hence, we have \vspace{-0.15cm}
\begin{align}\label{eq:LemmaAttrEVProb}\begin{split}
&\Pr \left\{  \exists n : | B_{n}(t) - \varepsilon_{n}^\star| > ((S+1)D+{m})/ \C \right\} \leq \\
&\quad \Pr\left\{ \exists (n,s): |U_{n,s}(t) - \nu_{n,s}^\star| >D+{m\over S+1} \right\} +  \Pr \left\{\exists n: |Z_{n}(t) - \zeta_{n}^\star| >D+{m\over S+1}\right \}.\end{split}
\end{align}
Summing both sides of \eqref{eq:LemmaAttrEVProb} over $t$, taking limit superior and using subadditivity property of $\limsup$, we obtain\vspace{-0.15cm}
\begin{align}\label{eq:LemmaAttrEVProb3}\begin{split}
&\limsup_{T \rightarrow \infty} \sum_{t=0}^{T-1} \Pr \{   \exists n : | B_{n}(t) - \varepsilon_{n}^\star| > ((S+1)D+{m})/ \C\} \leq\\
&\quad \limsup_{T \rightarrow \infty} \sum_{t=0}^{T-1} \Pr \{ \exists (n,s): |U_{n,s}(t) - \nu_{n,s}^\star| >D+{m\over S+1} \} +  \\
&\quad \limsup_{T \rightarrow \infty} \sum_{t=0}^{T-1}\Pr  \{\exists n:|Z_{n}(t) - \zeta_{n}^\star| >D+{m\over S+1}\} \leq{2c^\star e^{-\beta^\star {m\over S+1}}},\end{split}
\end{align}
where the last inequality holds due to the upper bounds  of \eqref{eq:ThAttrCuP} and \eqref{eq:ThAttrCz}. This completes the proof of Theorem \ref{sec:attrTheorem}. 
\vspace{-0.15cm}
\section{Proof of Lemma \ref{lem:limitedBattery}} \label{sec:appLimBB}

{Implementing  \alg{} in the cases with a limited battery and buffers, data is dropped due to  either low buffer space or energy outage. Let $L_{n,s}^b(t)$ and $L_{n,s}^e(t)$ denote the total number of data units that are dropped up to time  $t$ due to   low buffer space and  energy outage, respectively.  We have\vspace{-0.15cm}}
\begin{align}\label{eq:proofLB_deCom}
L_{n,s}(t) = L_{n,s}^b(t)+L_{n,s}^e(t).
\end{align}
{Moreover, let $\mathcal{M}_{n,s}(t)$ denote the total number of data units of stream $s$ that have entered  node $n$ up to time $t$ while the level of   $\tilde{U}_{n,s}(t)$ has been outside the  interval  of  length $\Uc$ around $\nu^\star$. We have\vspace{-0.15cm}}
\begin{align}\label{eq:proofLB_Lb_Bound}
\mathcal{M}_{n,s}(t) \leq \sum_{\tau=0}^t\uitau  \iif{|\tilde{U}_{n,s}(\tau) - \nu^\star| \geq \Uc /2 -\uimax},
\end{align}
where $\iif{x\geq a} = 1$ if $x \geq a$ and $ 0$ otherwise, for $x,a \in \mathbb{R}$. We claim that\vspace{-0.15cm}
\begin{align}\label{eq:proofLBClaim}
L_{n,s}^b(t) \leq \mathcal{M}_{n,s}(t) \leq\sum_{\tau=0}^t\uitau  \iif{|\tilde{U}_{n,s}(\tau) - \nu^\star| > \Uc /2 -\uimax} .\vspace{-0.15cm}
\end{align}
To prove \eqref{eq:proofLBClaim}, suppose that a genie-aided dropping discipline is used for dropping data units when overflow occurs.  The genie-aided algorithm  selects the data units to be dropped arbitrarily  among those data  units that  have  entered  node $n$ when the level of   $\tilde{U}_{n,s}(t)$ has been outside the  interval  of  length $\Uc$ around $\nu^\star$. With  the genie-aided algorithm, we have $L_{n,s}^b(t) < \mathcal{M}_{n,s}(t)$ since only the data units counted under $ \mathcal{M}_{n,s}(t)$ are candidates to be dropped.  The number of dropped data units is independent of the discipline for selecting  the dropped data units. Hence  \eqref{eq:proofLBClaim} always holds independently of the algorithm for dropping the data.

  Let $\hat{L}_{n}^e(t)$ denote the amount of energy that is not stored in the battery of node $n$ due to energy overflow. Using a similar argument for deriving \eqref{eq:proofLBClaim}, $\hat{L}_{n}^e(t)$  is  bounded by\vspace{-0.15cm}
\begin{align}\label{eq:proofLB_Le_Bound}
\hat{L}_{n}^e(t)  \leq \sum_{\tau=0}^t \eitau  \iif{|E_{n}(\tau) - \varepsilon^\star_n| > \Ec /2 -\eimax  }.
\end{align}
From \eqref{eq:rateBound}, we have $L_{(n,s)}^e(t)  \leq \delta \hat{L}_{n}^e(t)$ and, consequently,\vspace{-0.15cm}
\begin{align}\label{eq:proofLB_Le_Bound2}
L_{(n,s)}^e(t)  \leq \delta \sum_{\tau=0}^t \eitau  \iif{|E_{n}(\tau) - \varepsilon^\star_n| > \Ec /2 -\eimax  }.
\end{align} 
Taking time average expected  value of both sides in \eqref{eq:proofLB_deCom} and using the upper-bounds in \eqref{eq:proofLBClaim} and \eqref{eq:proofLB_Le_Bound2}, we obtain
\begin{small}
\begin{align}\label{eq:proofLBLast}\begin{split}
\limsup_{T\rightarrow \infty }& {1\over T} \eexp{L_{n,s}(T)} \leq \limsup_{T\rightarrow \infty } {1\over T}\sum_{\tau=0}^T\uimax  \eexp{\iif{|U_{n,s}(\tau) - \nu^\star| > \Uc /2 -\uimax}} + \\
&\qquad\limsup_{T\rightarrow \infty } {1\over T} \sum_{\tau=0}^T\delta \eimax  \eexp{\iif{|E_{n}(\tau) - \varepsilon^\star_n| > \Ec /2 -\eimax  }}\\
&\overset{(\text{a})}{\leq} \limsup_{T\rightarrow \infty } {1\over T}\sum_{\tau=0}^T\uimax  \Pr\{|U_{n,s}(\tau) - \nu^\star| > D+m_l\} + \\
&\qquad\limsup_{T\rightarrow \infty } {1\over T} \sum_{\tau=0}^T \delta\eimax  \Pr\{|E_{n}(\tau) - \varepsilon^\star_n| > ((S+1)D+m_l)/\C\}\overset{(\text{b})}{\leq}  {\uimax c^\star e^{-\beta^\star m_l} } + {2\delta\eimax c^\star e^{-\beta^\star {m_l\over S+1}} },\end{split}
\end{align}
\end{small}
where $(\text{a})$ holds since $\eexp{\iif{x>a}} = \Pr\{x>a\}$ and from \eqref{eq:m_lDef}  we have $D+m_l \leq \Uc /2 -\uimax$ and $((S+1)D+m_l)/\C\leq \Ec /2 -\eimax$. The last inequality $(\text{b})$ results from the upper-bounds of \eqref{eq:ThAttrCu} and \eqref{eq:ThAttrCe} with $m = m_l$.

\bibliographystyle{IEEEtran}
\bibliography{paper}

\end{document}